\ifdefined\fullversion
\else
\def\fullversion{1}    
\fi

\ifdefined\cameraversion
\else
\def\cameraversion{0}    
\fi

\def\anonymous{0}      

\documentclass[envcountsame,runningheads,notitlepage]{llncs}
\ifnum\fullversion=1
\usepackage[a4paper, margin=1.1in]{geometry}
\fi

\usepackage[utf8]{inputenc}
\usepackage[T1]{fontenc}
\usepackage{verbatim}
\usepackage{tikz}
\usetikzlibrary{decorations.pathreplacing}
\usetikzlibrary{positioning,calc}
\usepackage{xspace}
\usepackage{amsmath} 
\usepackage{amssymb}
\usepackage{mathtools}
\usepackage{pifont}
\usepackage{etoolbox}
\usepackage[normalem]{ulem}
\usepackage{booktabs}
\usepackage{array}
\usepackage[hidelinks]{hyperref} 
\usepackage[capitalise,noabbrev]{cleveref}
\usepackage{cite}
\usepackage{url}
\usepackage[noend]{algpseudocode}
\usepackage[ruled,vlined,linesnumbered]{algorithm2e}
\usepackage{paralist}
\usepackage{mathrsfs}
\usepackage{relsize}
\usepackage{stmaryrd}
\usepackage[textsize=small]{todonotes}
\usepackage{multirow}
\usepackage[lambda,n,operators]{cryptocode}
\usepackage{outlines}
\usepackage{braket}

\newtoggle{notes}
\toggletrue{notes} 

\newcommand{\SEnote}[1]{{\color{red} \footnotesize(Simona: #1)}}
\newcommand{\JLnote}[1]{{\color{cyan} \footnotesize(Johanna: #1)}}
\newcommand{\JLLnote}[1]{{\color{cyan!70!black} \footnotesize(#1)}}
\newcommand{\LEnote}[1]{{\color{blue} \footnotesize(Lynn: #1)}}
\newcommand{\ALLnote}[1]{{\color{orange} \footnotesize(All: #1)}}

\newcommand{\DeleteAllComments}{
    \renewcommand{\SEnote}[1]{}
    \renewcommand{\JLnote}[1]{}
    \renewcommand{\LEnote}[1]{}
    \renewcommand{\ALLnote}[1]{}
    \renewcommand{\JLLnote}[1]{}
}

\counterwithin{theorem}{section}
\counterwithin{lemma}{section}
\counterwithin{corollary}{section}
\counterwithin{proposition}{section}
\counterwithin{definition}{section}
\counterwithin{remark}{section}
\counterwithin{problem}{section}

\newcommand{\oracleA}{\ensuremath{\mathcal{A}}\xspace}

\newcommand{\F}{\mathbb{F}}
\newcommand{\N}{\mathbb{N}}
\newcommand{\R}{\mathbb{R}}

\newcommand{\xv}{\mathbf{x}}
\newcommand{\yv}{\mathbf{y}}
\newcommand{\cv}{\mathbf{c}}

\newcommand{\vv}{\mathbf{v}}
\newcommand{\wv}{\mathbf{w}}

\newcommand{\Gm}{\mathbf{G}}

\newcommand{\Ss}{\mathcal{S}}
\newcommand{\Cs}{\mathcal{C}}
\newcommand{\Ws}{\mathcal{W}}

\newcommand{\Sf}{\mathsf{S}}
\newcommand{\Uf}{\mathsf{U}}
\newcommand{\Cf}{\mathsf{C}}

\newcommand{\Region}{\mathrm{Region}}

\newcommand{\NNS}{\textsc{NNS}}
\newcommand{\DP}{\textsc{DP}}

\let\oldnl\nl
\newcommand{\nonl}{\renewcommand{\nl}{\let\nl\oldnl}}
\newcommand{\subr}[1]{\textnormal{\texttt{#1}}}

\newcommand{\Input}{\ensuremath{\mathsf{Input}}\xspace}

\makeatletter
\DeclareFontFamily{OMX}{MnSymbolE}{}
\DeclareSymbolFont{MnLargeSymbols}{OMX}{MnSymbolE}{m}{n}
\SetSymbolFont{MnLargeSymbols}{bold}{OMX}{MnSymbolE}{b}{n}
\DeclareFontShape{OMX}{MnSymbolE}{m}{n}{
    <-6>  MnSymbolE5
   <6-7>  MnSymbolE6
   <7-8>  MnSymbolE7
   <8-9>  MnSymbolE8
   <9-10> MnSymbolE9
  <10-12> MnSymbolE10
  <12->   MnSymbolE12
}{}
\DeclareFontShape{OMX}{MnSymbolE}{b}{n}{
    <-6>  MnSymbolE-Bold5
   <6-7>  MnSymbolE-Bold6
   <7-8>  MnSymbolE-Bold7
   <8-9>  MnSymbolE-Bold8
   <9-10> MnSymbolE-Bold9
  <10-12> MnSymbolE-Bold10
  <12->   MnSymbolE-Bold12
}{}

\let\llangle\@undefined
\let\rrangle\@undefined
\DeclareMathDelimiter{\llangle}{\mathopen}%
                     {MnLargeSymbols}{'164}{MnLargeSymbols}{'164}
\DeclareMathDelimiter{\rrangle}{\mathclose}%
                     {MnLargeSymbols}{'171}{MnLargeSymbols}{'171}
\makeatother

\title{Quantum Sieving for Code-Based Cryptanalysis \\ and Its Limitations for ISD}
\titlerunning{Quantum Sieving for Code-Based Cryptanalysis and Its Limitations for ISD}
\date{}

\ifnum\anonymous=0
\author{
  Lynn Engelberts\inst{1,2} \and
  Simona Etinski\inst{1} \and 
  Johanna Loyer\inst{1} \\ {\small \texttt{firstname.lastname@cwi.nl}} 
} 

\institute{CWI, The Netherlands \and QuSoft, The Netherlands} 

\else
\author{} 
\institute{}
\fi

\authorrunning{Lynn Engelberts, Simona Etinski, and Johanna Loyer} 

\begin{document}
\maketitle

\DeleteAllComments 

\begin{abstract}
Sieving using near-neighbor search techniques is a well-known method in lattice-based cryptanalysis, yielding the current best runtime for the \text{shortest vector problem} in both the classical \cite{BDGL16} and quantum \cite{BCSS23} setting. 
Recently, sieving has also become an important tool in code-based cryptanalysis. Specifically, using a sieving subroutine, \cite{EPRINT:GuoJohNgu23,DEEK23} presented a variant of the information-set decoding (ISD) framework, which is commonly used for attacking cryptographically relevant instances of the decoding problem. 
The resulting sieving-based ISD framework yields complexities close to the best-performing classical algorithms for the decoding problem such as \cite{EC:BJMM12, BM18}. It is therefore natural to ask how well quantum versions perform. 
In this work, we introduce the first quantum algorithms for code sieving by designing quantum variants of the aforementioned sieving subroutine. In particular, using quantum-walk techniques, we provide a speed-up over the best known classical algorithm from \cite{DEEK23} and over a variant using Grover's algorithm. Our quantum-walk algorithm exploits the structure of the underlying search problem by adding a layer of locality-sensitive filtering, inspired by the quantum-walk algorithm for lattice sieving from \cite{AC:ChaLoy21}. We complement our asymptotic analysis of the quantum algorithms with numerical results, and observe that our quantum speed-ups for code sieving behave similarly as those observed in lattice sieving. 
In addition, we show that a natural quantum analog of the sieving-based ISD framework does not provide any speed-up over the first presented quantum ISD algorithm \cite{Bernstein10}. Our analysis highlights that the framework should be adapted in order to outperform the state-of-the-art of quantum ISD algorithms \cite{PQCRYPTO:KacTil17, PQCRYPTO:Kirshanova18}. 
\end{abstract}

\paragraph{Keywords.}
Quantum cryptanalysis, 
Quantum walks, 
Near-neighbor search, 
Code sieving, 
Decoding problem, 
Information-set decoding

\section{Introduction}
\label{sec: introduction}
A fundamental problem in code-based cryptography is the decoding problem: given a linear code $\mathcal{C}$, find a codeword $\mathbf{x}_c \in \mathcal{C}$ of (small) fixed weight $w$.\footnote{This problem is called the codeword-finding problem in \cite{DEEK23}, and can be seen as a homogeneous version of the well-known syndrome decoding problem.} The decoding problem is NP-hard in the worst case \cite{BMvT78}. 
More important for cryptographic purposes, it is known that $w$ can be chosen to guarantee the existence of exactly one solution on average, which gives us the so-called \textit{unique decoding regime}. In this regime, the decoding problem is believed to be hard for a random instance of the problem. 
In particular, all known algorithms for attacking a random instance of this problem run in time and memory exponential in the input size. The best generic attacks\footnote{For certain parameter regimes, the so-called statistical decoding attacks (also known as dual attacks in the lattice literature) are more efficient. As we are not particularly interested in a specific parameter regime, we will only consider attacks in the ISD framework.} belong to the so-called \textit{information-set decoding} (ISD) framework. 

The ISD framework was originally introduced in the work of Prange \cite{Pra62}, and further improved using various techniques, including the meet-in-the-middle approach and its generalizations \cite{Ste88, Dum91, SS81}, representation techniques \cite{AC:MayMeuTho11, EC:BJMM12}, near-neighbor search techniques \cite{EC:MayOze15, Car20}, and their combinations. Some of these techniques were adapted to the quantum setting \cite{Bernstein10, PQCRYPTO:KacTil17, PQCRYPTO:Kirshanova18}.
To find the solution to the decoding problem, current ISD algorithms search for many partial solutions (by solving an instance of the decoding problem in a smaller dimension), and then check if any of those yields a solution to the original problem. 

Recently, a new ISD algorithm was proposed in \cite{EPRINT:GuoJohNgu23} and further generalized and improved in \cite{DEEK23}. 
This approach uses a so-called \textit{sieving technique} to find the partial solutions, which is well-known and widely applied in lattice-based cryptanalysis. 
The idea of a sieving algorithm is to start with a list of arbitrary elements and iteratively combine elements from the current list to obtain a new list  of elements that satisfy a certain property which makes them (in some precise mathematical sense) `closer' to being a solution. 
At the end, the algorithm outputs a list of solutions. 
In the code setting, the sieve starts with a list of arbitrary vectors, and in each iteration it combines pairs of vectors that satisfy an additional code constraint, eventually ending up with vectors in the code. 

In each iteration of the sieve, the task of combining list elements to obtain a new list of suitable pairs can be formulated as an instance of a near-neighbor search problem (NNS). \cite{DEEK23} presented several methods for solving this NNS problem, the best of which are based on locality-sensitive filtering (LSF) techniques. 
Using these methods, \cite{DEEK23} obtained a sieving-based ISD algorithm for the decoding problem whose asymptotic runtime is close to that of the best known classical algorithms (such as \cite{EC:BJMM12, BM18}), with an improved time-memory trade-off over the previously known techniques. 

While sieving algorithms are new in code-based cryptanalysis, they are abundant in lattice-based cryptanalysis, and belong to the state-of-the-art of classical \emph{and} quantum algorithms for the shortest-vector problem. 
Classical algorithms for lattice sieving were first introduced by \cite{AKS01} and later improved in \cite{NV08, MV10, BDGL16}, the latter of which introduced the LSF techniques in the Euclidean metric that were adapted in \cite{DEEK23} for the Hamming metric. 
The addition of LSF gave a significant improvement in the runtime of lattice sieving. The first quantum speed-up for lattice sieving was obtained using a Grover-based algorithm \cite{Laa16}, which was further improved using quantum-walk techniques in \cite{AC:ChaLoy21}. This quantum-walk algorithm was further improved in \cite{BCSS23} using reusable-walk techniques.

The various quantum speed-ups in lattice sieving raise a natural question of how well the corresponding quantum techniques perform when adapted to the setting of codes, which is the focus of this work. 
More precisely, in this work we aim to answer the following questions.  
\begin{enumerate}
    \item What is (an upper bound on) the quantum complexity of code sieving?
    \item What is the runtime and memory of the quantum ISD algorithm that results from using our quantum algorithms for code sieving as a subroutine?
\end{enumerate}

\subsection{Main Contributions}

\paragraph{Quantum algorithms for code sieving.} 

We introduce the first quantum algorithms for code sieving by combining the classical sieving algorithm from \cite{DEEK23} with the state-of-the-art techniques from quantum algorithms for lattice sieving \cite{Laa16,AC:ChaLoy21, BCSS23}. 
The sieving algorithm in \cite{DEEK23} appears as a subroutine of their ISD algorithm, and solves a decoding problem with the goal of finding many codewords in a given code. 
This subroutine repeatedly solves a near-neighbor search problem of the following form: given a list $\mathcal{L}$ of $N$ vectors in $\F_2^n$ of (Hamming) weight $w$, find all pairs in $\mathcal{L}$ whose sum is of weight $w$. 
The complexity of the best known algorithms for this NNS problem, and thereby the overall complexity of the sieving-based algorithm, depends on the cost of a subroutine that we call $\subr{FindSolutions}$ which, in short, searches for `solutions' in a structured subset of the list $\mathcal{L}$. Similar to the lattice setting, our quantum algorithms aim to speed up this subroutine. We thus obtain quantum algorithms (and speed-ups) for this NNS problem, which might be of independent interest. 

We present several quantum algorithms for $\subr{FindSolutions}$ and analyze their asymptotic time and memory complexities. 
Our first quantum algorithm is a straightforward application of Grover's algorithm, and serves as a baseline for comparison. 
Our other quantum algorithms are a quantum-walk algorithm and a variant thereof, where we apply the sparsification technique from \cite{AC:ChaLoy21}. 
Both quantum-walk algorithms make use of locality-sensitive filtering (LSF) to obtain a speed-up over Grover in the search for `solutions'. 
The use of LSF inside the quantum walk was first suggested by \cite{AC:ChaLoy21} in the context of lattice sieving, and one of our contributions is to show how to apply this layer of LSF in the context of codes (see \Cref{sec: reduction to a variant of NNS}). Specifically, we introduce a Hamming-metric variant of the `residual vectors' used in \cite{AC:ChaLoy21}, which behave somewhat differently than their Euclidean-metric analogs, but still allow us to apply LSF.  
 
Besides providing an asymptotic analysis of the time and memory complexities of our quantum algorithms, we evaluate their performance numerically. Namely, the complexity of our quantum algorithms depends on some parameters that we  optimize numerically to obtain the smallest runtime (and corresponding memory) for each of the algorithms.   
Our numerical results illustrate that we obtain a quantum speed-up over the classical algorithm from \cite{DEEK23}, and provide a comparison between our different quantum algorithms. 
In addition, we show that the quantum speed-ups we obtain align with the speed-ups observed in the state-of-the-art of lattice cryptanalysis (see Table~\ref{tab:HardestInstances}). The Python code used for the numerical results of this work is available at \url{https://github.com/lynnengelberts/quantum-sieving-for-codes-public}.

\paragraph{Application to quantum ISD.}
The sieving framework from \cite{EPRINT:GuoJohNgu23, DEEK23} was introduced as a subroutine in an ISD algorithm to solve the decoding problem in the unique decoding regime.
Since the resulting sieving-based ISD algorithm is shown to asymptotically perform nearly as good as the classical state-of-the-art for this problem, we focus on the question whether quantum analogs of this algorithm could have similar performance as the quantum state-of-the-art. 
In particular, we consider a natural quantum analog of sieving-based ISD that is obtained by allowing for a \textit{quantum} sieving subroutine and using amplitude amplification. 
We show, using a combination of analysis and numerical optimization, that this quantum analog of sieving-based ISD \textit{cannot} even improve on the first quantum algorithm for the decoding problem from \cite{Bernstein10}. 

More precisely, we observe that the main limiting factor is a lower-bound constraint on the list size that is imposed in each iteration of the sieve, and which results in a lower bound on the time complexity of the resulting quantum ISD algorithm. 
Our results indicate that the sieving-based ISD framework should be adapted if it wants to compete with the best-performing quantum algorithms for the decoding problem.
We suggest two natural attempts to adapt the sieving-based ISD framework, and explain why these do not allow for overcoming the found limitations.  

\bigskip
In the end, our results show that code-based cryptosystems are still resilient to these quantum methods for code sieving inside the standard ISD framework. Given the quantum speed-ups in lattice-based cryptanalysis and the recent introduction of sieving techniques for codes, it was essential to evaluate their impact on the security of code-based schemes.
Our new quantum algorithms for the NNS problem in the Hamming metric have the potential to be used in future algorithms for the decoding problem (for instance within new sieving algorithms that overcome the limitations addressed in this work, or within algorithms with completely different approaches). Identifying specific applications remains an open question for future work. 

\subsection{Outline}
The paper is organized as follows. 
In \Cref{sec: preliminaries}, we introduce our notation and the relevant preliminaries. 
\Cref{sec: sieving for codes} describes the framework for code sieving, as well as near-neighbor search methods using locality-sensitive filtering. It also explains how the complexity of NNS and code sieving depends on the complexity of $\subr{FindSolutions}$.
In \Cref{sec: quantum NNS}, we present our quantum algorithms for $\subr{FindSolutions}$ and their asymptotic complexity. 
\Cref{sec: numerical results} presents numerical results for the asymptotic runtime and memory of the introduced quantum algorithms. 
Finally, in \Cref{sec: sievingISD}, we adapt the classical sieving-based ISD framework from \cite{DEEK23} to the quantum setting, and discuss its limitations.

\subsection{Acknowledgements}
The authors are grateful to Ronald de Wolf, L\'{e}o Ducas, and Elena Kirshanova for useful comments on the manuscript. The authors also thank Jean-Pierre Tillich for discussing this project in its early stage. 
LE was supported by the Dutch National Growth Fund (NGF), as part of the Quantum Delta NL program. 
SE and JL were supported by the ERC Starting Grant 947821 (ARTICULATE).

\section{Preliminaries}
\label{sec: preliminaries}
\subsection{Notation}

\paragraph{Binary finite fields.}
We denote by $\F_2$ the binary finite field, and by $\F_2^n$ the corresponding vector space of dimension $n$. We write $+$ (resp. $\land$) for the bitwise `XOR' (resp. `AND') between two vectors in $\mathbb F_2^n$.

\paragraph{Scalars, vectors, matrices.} 
A scalar is denoted by a non-bold, small letter, a vector is denoted by a bold small letter, and a matrix by a bold capital letter.

\paragraph{Asymptotic notation.}
We use standard Landau notation. 
Namely, we write $f(n) = O(g(n))$ if there exist constants $c, n_0 \geq 0$ such that $f(n) \leq c \cdot g(n)$ for all integers $n \geq n_0$. We write  $f(n) = \Omega(g(n))$ if there exist constants $c, n_0 \geq 0$ such that $f(n) \geq c \cdot g(n)$ for all integers $n \geq n_0$. We write $f(n) = \Theta(g(n))$ if both $f(n) = O(g(n))$ and $f(n) = \Omega(g(n))$. 
We define $\tilde{O}(f(n)) := O\left(f(n) \cdot \mathrm{polylog}(f(n))\right)$, where $\mathrm{polylog}(f(n)) := \log(f(n))^{O(1)}$, and define  $\tilde{\Omega}$, $\tilde{\Theta}$ in a similar way.  
We write $f(n) = o(g(n))$ if for all constants $c > 0$, there exists $n_0 > 0$ such that $0 \leq f(n) < c \cdot g(n)$ for all integers $n \geq n_0$.

\paragraph{Other.} 
The non-negative integers are denoted by $\N$. For any $n, k \in \N$ with $k \leq n$, we define the binomial coefficient by $\binom{n}{k} := \frac{n!}{k!(n-k)!}$. 

\subsection{Hamming Space}\label{sec: prelims Hamming space}

We are interested in computational problems over $\F_2^n$ endowed with the \textit{Hamming metric} for $n \in \N$. In particular, we define the \textit{Hamming weight} $|\cdot|$ of a vector as \[\forall \xv \in \F_2^n, \quad |\xv| := \{i \in \{0,1,\dots,n-1\} \colon x_i \neq 0\}.\] 
(The notation $|\cdot|$ will also be used to define the cardinality of a set.)

\begin{definition}[Hamming sphere]
For any integer $0 \leq w \leq n$, we define the \textit{weight-$w$ (Hamming) sphere} as 
\[\Ss_w^n := \{ \xv \in \F_2^n \colon |\xv| = w\}.\]
\end{definition}
The surface area of a sphere (i.e., the size of the set $\Ss_w^n$) is calculated as $|\Ss_w^n| = {n \choose w}$. 

\begin{definition}[Region] For any $\cv \in \F_2^n$ and integer $0 \leq \alpha \leq |\cv|$, we define 
    \[\Region_\alpha^n(\cv) := \{\xv \in \F_2^n \colon |\xv \land \cv| = \alpha\}.\]
\end{definition}
The surface area of a region is $|\Region_\alpha^n(\cv)| = \sum_{v = \alpha}^n {|\cv| \choose \alpha}  {n - |\cv| \choose v - \alpha}$. 

The intersection of a sphere with a region is called a \textit{(spherical) cap}. 
\begin{definition}[Spherical cap]
For any $\cv \in \F_2^n$ and integers $\alpha, v$ with $0 \leq \alpha \leq v \leq n$ and $\alpha \leq |\cv|$, we define
    $$\Cs_{v, \alpha}^n(\cv) = \Ss_v^n \cap \Region_\alpha^n(\cv) := \{\xv \in \Ss_v^n \colon |\xv \land \cv| = \alpha\}.$$
\end{definition}
The surface area of a spherical cap is $|\Cs_{v,\alpha}^n(\cv)| = |\Ss_v^n \cap \Region_\alpha^n(\cv)| = {|\cv| \choose \alpha}  {n - |\cv| \choose v - \alpha}$. 

Furthermore, the intersection of two caps (on the same sphere) is called a \textit{(spherical) wedge}.
\begin{definition}[Spherical wedge] For any $\xv, \yv \in \F_2^n$ and integers $\alpha, v$ with $0 \leq \alpha \leq v \leq n$ and $\alpha \leq \min(|\xv|,|\yv|)$, we define 
\begin{align*}
    \Ws_{v,\alpha}^n(\xv,\yv) &:= \Cs_{v,\alpha}^n(\xv) \cap \Cs_{v,\alpha}^n(\yv) = \Ss_v^n \cap \Region_\alpha^n(\xv) \cap \Region_\alpha^n(\yv) \\
    &= \{\cv \in \Ss_v^n \colon |\xv \land \cv| = |\yv \land \cv| = \alpha\}.
\end{align*}      
\end{definition}

We are particularly interested in the case where $\xv$ and $\yv$ are of equal Hamming weight. Then the surface area of a spherical wedge is given by the following lemma. 

\begin{lemma}[Surface area of a wedge] \label{lem: surface area wedge}
Let $\xv, \yv \in \Ss_w^n$ and define $w^{\ast} := |\xv \land \yv|$. 
For all integers $\alpha, v$ with $0 \leq \alpha \leq v \leq n$ and $\alpha \leq w$, we have that 
    \begin{align*}
        |\Ws_{v, \alpha}^n(\xv, \yv)| = \sum_{e = \max(0, 2\alpha - v)}^{\min(\alpha, w^{\ast})} {w^{\ast} \choose e} {w - w^{\ast} \choose \alpha - e}^2 {n - 2w + w^{\ast} \choose v - 2\alpha + e}.
    \end{align*}
\end{lemma}
Note that $e$ ranges over all possible values of $|\xv \land \yv \land \cv|$. 

\begin{proof}
    (The case $w^{\ast} = \frac{w}{2}$ was already proven in \cite{DEEK23}.) 
    We want to count the number of $\cv \in S_v^n$ such that $|\xv \land \cv| = |\yv \land \cv| = \alpha$. For each such vector $\cv$, it holds that $\max(0, 2\alpha - v) \leq |\xv \land \yv \land \cv| \leq \min(\alpha, w^{\ast})$, where the inequality $2\alpha - v \leq |\xv \land \yv \land \cv| $ follows from the fact that $v = |\cv|$ must be at least $|\xv \land \cv \land \yv| + |\xv \land \cv \land \overline{\yv}| + |\yv \land \cv \land \overline{\xv}| = |\xv \land \yv \land \cv| + 2(\alpha - |\xv \land \yv \land \cv|) = 2\alpha - |\xv \land \yv \land \cv|$ (where we write $\overline{\xv}$ for the bitwise negation of $\xv$). 
    In particular, we have the partition  \begin{align*}
        \Ws_{v, \alpha}^n(\xv, \yv) = \bigsqcup_{e = \max(0, 2\alpha - v)}^{\min(\alpha, w^{\ast})} S_e
    \end{align*}
    where $S_e := \{ \cv \in \Ss_v^n \colon |\xv \land \cv| = |\yv \land \cv| = \alpha, |\xv \land \yv \land \cv| = e\}$
    (note that the sets $S_e$ are clearly disjoint). 
    The claim then follows from the observation that $|S_e| = {w^{\ast} \choose e} {w - w^{\ast} \choose \alpha - e}^2 {n - (2w - w^{\ast}) \choose v - 2\alpha + e}$. 
    \qed 
\end{proof}

\begin{remark}[Surface area only depends on the center weight] \label{rem: surface area only depends on weight}
    Note that the surface area of a cap (resp.\ wedge) only depends on the \textit{weight} of the center $\cv$ (resp.\ on the weight of $\xv, \yv$ and on their overlap).
\end{remark}

\subsection{Linear Codes and Random Codes}\label{sec: prelims linear codes}

The problems and algorithms in this work consider binary linear codes. 

\begin{definition}[Binary linear code]
An $[n,k]$ binary linear code $\Cs$ is defined as a linear subspace of size $2^k$ of $\F_2^n$. Elements of $\Cs$ are called \textit{codewords}. 
\end{definition}

A code can be represented by a \textit{generator matrix}: a full-rank matrix $\Gm \in \F_2^{k \times n}$ whose rows form a basis of the subspace $\Cs$. Conversely, any such matrix defines a code. 

We are primarily interested in \textit{random binary linear codes} (abbr.\ \textit{random codes}), namely codes obtained from sampling a generator matrix uniformly at random from $\F_2^{k \times n}$. These codes can be sampled in time polynomial in $n$.

\subsection{Quantum Computation}\label{sec: quantum prelims}

We say that a quantum algorithm has time complexity $T$ if the circuit describing it uses at most $T$ elementary quantum gates. When $T$ is polynomial in the number of bits needed to write down the input, we say the algorithm is \textit{efficient}. 
In the following, we define $[N] := \{0,\ldots, N-1\}$. 

\subsubsection{QRAM model and complexities.}
Our results hold in the quantum circuit model where QRAM operations are assumed efficiently implementable. We consider separately the quantum random access to classical memory (QRACM) and quantum random access to quantum memory (QRAQM). More precisely, a QRACM operation corresponds to an application of a unitary of the form $O_x \colon \ket{i}\ket{0} \mapsto \ket{i}\ket{x_i}$ for $x = (x_0, \dots, x_{N-1}) \in \{0,1\}^N$ and $i \in [N]$. We define a QRAQM operation to be an application of a unitary of the form $O' \colon  \ket{x_0, \ldots , x_i, \ldots, x_{N-1}}   \ket{i} \ket{b}  \mapsto \ket{x_0, \ldots , b, \ldots, x_{N-1}} \ket{i} \ket{x_i} $ for  $x = (x_0, \ldots, x_{N-1}) \in \{0,1\}^N$, $i \in [N]$, and $b \in \{0,1\}$. 
We emphasize that the latter operation allows to read and to write in superposition, and that it was already used in \cite{Amb07}.

For the memory complexities of our algorithms, we will specify the sizes of the classical and quantum memories, and the sizes of the classical and quantum registers that the respective QRAM operations use. (That is, for an algorithm using the previous QRAM operations, we say it uses QRACM, respectively QRAQM, of size $N$.)  

\subsubsection{Grover's algorithm and amplitude amplification.} We recall Grover's algorithm and a generalization of it, commonly known as amplitude amplification. 

\begin{theorem}[Grover's algorithm] \label{thm: Grover}
    For $N \in \N$, let $f \colon [N] \to \{0,1\}$ be a function and define $t := |\{i \in [N]: f(i) = 1\}|$. 
    There exists a quantum algorithm, called \emph{Grover's algorithm}, that returns $i$ such that $f(i)=1$ with probability at least $2/3$, if such an $i$ exists, using $O(\sqrt{N/t})$ queries to $f$ and $\Tilde{O}(\sqrt{N/t})$ elementary gates. 
    The algorithm uses $\lceil \log_2(N) \rceil$ quantum memory.
\end{theorem} 

Grover's algorithm as originally presented in \cite{Gro96} applies to the case of a unique solution; the extension to multiple solutions was detailed in \cite{BBHT98}. When $t$ is known, then Grover's algorithm can be adapted to have success probability $1$. 
When $t$ is unknown, then it is still possible to find $i\in [N]$ such that $f(i) = 1$ with an expected number of $O(\sqrt{N/t})$ queries to $f$ (and an expected number of $\Tilde{O}(\sqrt{N/t})$ elementary gates) if $t > 0$. In this work, when we apply Grover's algorithm we always know an estimate of the number of solutions $t$.

\begin{theorem}[Amplitude amplification \cite{BHMT02}]\label{thm: AA}
    For $N \in \N$, let $f \colon [N] \to \{0,1\}$ be a function. Suppose there is an efficient quantum circuit $U_f$ that maps $\ket{i} \to (-1)^{f(i)} \ket{i}$. 
    Suppose $\mathcal{A}$ is an algorithm that can be implemented as a reversible quantum circuit and that returns  $i \in [N]$ such that $f(i) = 1$ with success probability $p > 0$ in time $T$. 
    Then there exists a quantum algorithm, called \emph{amplitude amplification}, that returns $i \in [N]$ such that $f(i) = 1$ with probability at least $\max(1 - p, p)$ in time $O(1/\sqrt{p}) \cdot T$. 
\end{theorem}

Moreover, for known $p$, there is a variant of amplitude amplification with success probability 1, which uses $O(1/\sqrt{p})$ applications of $\mathcal{A}$ and $U_f$ \cite[Theorem 4]{BHMT02}. 
When $p$ is unknown, then there is a variant that finds an $i\in [N]$ such that $f(i) = 1$ using an expected number of $\Theta(1/\sqrt{p})$ applications  of $\mathcal{A}$ and $U_f$ as long as $p>0$ \cite[Theorem 3]{BHMT02}.  

Finally, note that we can reduce the error probability of Grover's algorithm and amplitude amplification using standard methods: using $O(\log_2(1/\epsilon))$ repetitions it can be bounded from above by an arbitrarily small $\epsilon > 0$. 

\subsubsection{Quantum walk.}
We will consider the quantum walks as presented in \cite{MNRS11}. A quantum walk starts with an undirected graph $G = (V,E)$, with $V$ the set of vertices and $E \subseteq V \times V$ the set of edges. The set $M \subseteq V$ contains elements said marked, and the goal of a quantum walk is to return a marked vertex $v \in M$. For any vertex $v \in V$, we define $N(x) := \{v' : (v,v') \in E\}$ the set of neighbors of $v$, and
$\ket{p_v} = \sum_{y \in N(v)} \frac{1}{\sqrt{|N(v)|}} \ket{v'}$. We now define the following quantities:
\begin{itemize}
	\item Set-up cost $\Sf$ is the cost of the unitary map $\ket{0} \mapsto \frac{1}{\sqrt{|V|}} \sum_{v \in V} \ket{v}\ket{p_v}$. 
	\item Update cost $\Uf$ is the cost of the unitary map 
	$\ket{v}\ket{0} \mapsto \ket{v}\ket{p_v}$. 
	\item Checking cost $\Cf$ is the cost of computing the function $f : V \rightarrow \{0, 1\}$ where $f(v) = 1 \Leftrightarrow v \in M$. 
    \item $\epsilon = \frac{|M|}{|V|}$ is the fraction of marked vertices. 
    \item $\delta$ is the spectral gap of $G$, which is defined as $\delta := 1 - |\lambda|$, where $\lambda$ is the second largest eigenvalue (in magnitude) of the normalized adjacency matrix of $G$. 
\end{itemize}

\begin{proposition} \cite{MNRS11}
There exists a quantum-walk algorithm that finds a marked element $v \in M$ in time 
	$$O\left( \Sf + \frac{1}{\sqrt{\epsilon}}\left(\frac{1}{\sqrt{\delta}}\Uf + \Cf\right) \right).$$
\end{proposition}

\subsubsection{Johnson graph.}
For positive integers $r < n$, the Johnson graph $J(n,r)$ is a standard graph to run quantum walks, and is for instance used in quantum-walk algorithms for solving collision problems. Each vertex $v$ consists of $r$ distinct (unordered) points $x_1, \dots, x_r \in [n]$ and some additional data $D(v)$. 
Two vertices $v = (x_1,\dots,x_r,D(v))$ and $v' = (x'_1,\dots,x'_r,D(v'))$ form an edge in $J(n,r)$ iff $|\{x_1,\dots,x_r\} \cap \{x'_1,\dots,x'_r\}|=r-1$. 
It is well-known that the graph $J(n,r)$ has spectral gap $\delta = \Omega(\frac{1}{r})$. 

\cite{Amb07} and \cite{BJLM13} presented quantum data structures that use efficient QRAM to perform efficient insertion and deletion of elements in quantum superposition. We refer the reader to these papers for more details.

\section{Code Sieving using NNS Techniques}
\label{sec: sieving for codes}
Sieving using a near-neighbor search (NNS) subroutine is a well-known method in lattice-based cryptanalysis (e.g., see \cite{BDGL16}) and recently has become an important tool in code-based cryptanalysis as well. 
Our quantum algorithms for code sieving are based on the sieving framework that was introduced in \cite{EPRINT:GuoJohNgu23} and further improved and generalized in \cite{DEEK23}. 
In this section, we recall this sieving framework and its complexity (Theorem~\ref{thm: from NNS to DP}). 
We begin by stating the computational problem that is tackled by the sieving framework, namely, the \textit{decoding problem}. 

\begin{problem}[Decoding problem, $\DP(n,k,w,N)$]
\label{prob: decoding problem, N sol}
    Given an $[n,k]$ binary linear code $\mathcal{C}$ and an integer value $w$, find $N$ codewords $\mathbf{x}_c \in \mathcal{C}$ of weight $|\mathbf{x}_c| := w$. 
\end{problem}  

We focus on the expected runtime of algorithms that solve a \textit{random} instance of Problem~\ref{prob: decoding problem, N sol}. That is, we assume that $\Cs$ is a random $[n,k]$ linear code. Given such a code $\Cs$, the expected number of codewords in $\Cs$ of weight $w$ is ${n \choose w}/2^{n-k}$. For the decoding problem $\DP(n,k,w,N)$ to be well-defined, we thus require $N = O\Big(\frac{{n \choose w}}{2^{n-k}}\Big)$. 

\begin{remark}
    In the introduction (and later in Section~\ref{sec: sievingISD}), we considered the decoding problem as Problem~\ref{prob: decoding problem, N sol} with $N = 1$, as the presumed hardness of this variant underlies the security of code-based cryptographic primitives. 
    The case for arbitrary $N$ naturally appears as a subinstance in the ISD framework, and in particular it is the problem that is tackled by the sieving subroutine in the ISD algorithms from \cite{EPRINT:GuoJohNgu23, DEEK23}. Therefore, in the remainder of this section, we consider this problem for arbitrary $N$. 
\end{remark}

\subsection{Framework for Code Sieving}

Recently, sieving-based algorithms were presented for this decoding problem  \cite{EPRINT:GuoJohNgu23, DEEK23}. These algorithms make use of an oracle that solves NNS in $\F_2^n$ endowed with Hamming metric, which is defined as follows. 

\begin{problem}[Near-neighbor search, $\NNS(n,w,N)$]\label{prob: NNS} 
Let $n, w \in \N$ with $w \leq n$.
Given a list $\mathcal{L}$ of $N$ vectors sampled independently and uniformly at random from $S_w^n$, find a $(1 - o(1))$-fraction of all pairs $(\xv, \yv) \in \mathcal{L}^2$ satisfying $|\xv + \yv| = w$ (called \textit{solution pairs}).
\end{problem}

The basic idea of code sieving \cite{EPRINT:GuoJohNgu23, DEEK23} for solving $\DP(n,k,w,N)$ is to start from a list of arbitrary vectors of Hamming weight $w$ and then iteratively add code constraints to obtain a required number of codewords of Hamming weight $w$. These iteratively added code constraints define a so-called \textit{tower of codes}: a collection $\{\mathbb F_2^{n} = \mathcal{C}_0, \dots, \mathcal{C}_{n-k} = \mathcal{C}\}$  of codes with $\mathcal{C}_{i-1} \supseteq \mathcal{C}_{i}$ and dimension decrements of 1, starting with the `initial' code $\F_2^n$ and ending with the input code $\mathcal{C}$.

When going from $\mathcal{C}_{i-1}$ to $\mathcal{C}_{i}$, the addition of one constraint results in halving the linear subspace $\mathcal{C}_{i-1}$. 
Thus, in expectation only half of the vectors from the code $\mathcal{C}_{i-1}$ are in the code $\mathcal{C}_{i}$. 
To avoid an exponential drop, Algorithm~\ref{alg: sieving} therefore instead creates new elements from the elements of the previous iteration. 
Specifically, this is done by adding pairwise sums of the elements from the previous list $\mathcal{L}_{i-1}$ to the new list $\mathcal{L}_i$, which is a technique also used in lattice sieving. To improve the algorithm's performance, instead of searching through all pairs of elements from  $\mathcal{L}_{i-1}$, the algorithm searches through the list of \textit{near neighbors}, $\{(\xv, \yv) \colon \xv,\yv \in \mathcal{L}_{i-1}, |\xv + \yv| = w\}$, namely the elements that are at small Hamming distance in $\mathbb F_2^n$. Among the near neighbors $(\xv, \yv)$, the algorithm then adds $\xv + \yv$  to $\mathcal{L}_i$ for those $\xv + \yv$ that belong to $\mathcal{C}_i$. 

Algorithm \ref{alg: sieving} presents the overall sieving framework for solving the decoding problem (Problem~\ref{prob: decoding problem, N sol}) in more detail. Note that the time complexity of this sieving algorithm (Algorithm~\ref{alg: sieving}) is determined by the cost of one iteration, i.e., by the cost of solving $\NNS(n,w,N)$ (Problem~\ref{prob: NNS}).  
We therefore state the time and memory complexity of Algorithm~\ref{alg: sieving} after presenting the state-of-the-art algorithm from \cite{DEEK23} for $\NNS(n,w,N)$ in \Cref{sec: NNS using LSF}.   

\begin{algorithm}[h!]
    \SetKwInOut{Input}{Input}\SetKwInOut{Output}{Output}
    \Input{$[n,k]$ binary linear code $\mathcal{C} \subseteq \mathbb F_2^n$, weight $w$, output size $N$, oracle $\mathcal{A}_{\NNS}$ for $\NNS(n,w,N)$}
    \Output{$\mathcal{L} \subseteq \mathcal{C} \cap \mathcal{S}_{w}^{n}$ of size $N$
    }  
    \vspace{0.2cm}
    \nonl \textsc{Initialization:} \\
    \vspace{0.1cm}
    Sample a tower of codes $\{\mathbb F_2^{n} = \mathcal{C}_0, \dots, \mathcal{C}_{n-k} = \mathcal{C}\}$, with dimension decrements of 1 \\
    Sample $N$ vectors independently and uniformly at random from $\mathcal{S}_{w}^{n}$ and add them to a list $\mathcal{L}_{0}$ \\
    \vspace{0.2cm}
    \nonl \textsc{Sieving part:} \\
    \vspace{0.1cm}
    \For{$i = 1$ to $n-k$}
    {
        Invoke $\mathcal{A}_{\NNS}$ on $\mathcal{L}_{i-1}$ to obtain $\mathcal{L}'_i := \{(\xv, \yv) \colon \xv,\yv \in \mathcal{L}_{i-1}, |\xv + \yv| = w\}$\\
        \For{$(\mathbf{x}, \mathbf{y}) \in \mathcal{L}'_i$}
        {
            \If{$\mathbf{x} + \mathbf{y} \in \mathcal{C}_i$}
            {
                Add $\mathbf{x} + \mathbf{y}$ to $\mathcal{L}_{i}$
            } 
        }
    Discard some elements if $|\mathcal{L}_{i}| > N$ 
    }
    \Return $\mathcal{L}_{n-k}$
    \caption{Code sieving using NNS}
    \label{alg: sieving}
\end{algorithm}

\begin{remark}[Difference with lattice sieving]
    Note the difference with lattice sieving, where one starts with a list of long lattice vectors, and iteratively combines them to obtain \emph{shorter} lattice vectors. Another difference is that code sieving was presented in \cite{EPRINT:GuoJohNgu23, DEEK23} as a subroutine in an ISD algorithm, whereas lattice sieving is the main algorithm. (We elaborate on the application to ISD in Section~\ref{sec: sievingISD}.)
\end{remark}

Besides the upper bound on the output size $N$ that is imposed by the decoding problem itself, the construction of the sieving algorithm also puts a lower bound on the size $N$ used in the algorithm. 
Specifically, Lemma~\ref{lem: lower bound on N for NNS} shows that if $N$ is too small, then the expected number of solutions to $\NNS(n,w,N)$ is significantly smaller than $N$. As a result, the current list size would shrink in each iteration, and the output list cannot be of size $\Omega(N)$.

\begin{lemma}[\cite{DEEK23}]\label{lem: lower bound on N for NNS}
    Let $\mathcal{L}$ be a set containing $N$ vectors sampled uniformly and independently at random from $S_w^n$. Then the expected number of pairs $(\xv, \yv) \in \mathcal{L}^2$ satisfying $|\xv + \yv| = w$ is \begin{align*}
        N^2 \frac{{w \choose w/2} {n-w \choose w/2}}{{n \choose w}}.
    \end{align*} 
    In particular, the expected number of such pairs is at least $N$ if and only if $N \geq \frac{{n \choose w}}{{w \choose w/2} {n-w \choose w/2}}$.
\end{lemma}

\begin{proof}
(This was also observed in \cite{DEEK23}.)
For $\xv, \yv \in \Ss_w^n$ sampled independently and uniformly at random, the probability that they satisfy $|\xv + \yv| = w$ is equal to ${w \choose w/2} {n-w \choose w/2}/{n \choose w}$.
To see this, one could use the fact that for any $\xv, \yv \in \Ss_w^n$, the condition $|\xv + \yv| = w$ is equivalent to the condition that $|\xv \land \yv| = w/2$. 
It follows that the expected number of pairs satisfying this condition is $N^2 {w \choose w/2} {n-w \choose w/2}/{n \choose w}$.
Therefore, the expected number of such pairs is at least $N$ if and only if $N \geq {n \choose w}\big/\big({w \choose w/2} {n-w \choose w/2}\big)$.
\qed
\end{proof} 

By Lemma~\ref{lem: lower bound on N for NNS}, if $N \geq {n \choose w}\big/\big({w \choose w/2} {n-w \choose w/2}\big)$, then $\NNS(n,w,N)$ has at least $N$ solutions in expectation. However, note that, in each iteration $i$ of Algorithm~\ref{alg: sieving}, on average only a quarter of the NNS solutions $(\xv, \yv)$ found by $\mathcal{A}_{\NNS}$ will be added to the new list $\mathcal{L}_i$. Indeed, half of them are discarded as each pair is counted twice: $(\xv, \yv)$ and $(\yv, \xv)$. Secondly, the algorithm discards those $\xv + \yv$ that do not belong to $\mathcal{C}_i$. 
Since two subsequent codes $\mathcal{C}_{i-1}, \mathcal{C}_{i}$ differ by one dimension, 
it follows that in expectation another half of the NNS solutions are discarded. 
Therefore, \cite{DEEK23} propose to take $N \geq 4 \cdot {n \choose w}\big/\big({w \choose w/2} {n-w \choose w/2}\big)$ to maintain a list size of $N$ through all iterations in the sieving part.

It is important to emphasize that the application of \Cref{lem: lower bound on N for NNS} in the above discussion is justified if the distribution of the list elements does not change throughout the iterations. 
In fact, to guarantee that the output list contains $N$ vectors (for $N$ within the stated bounds), the sieving algorithm from \cite{DEEK23} relies on the following heuristic.

\begin{quote}
    {\normalsize \textit{In each iteration of Algorithm~\ref{alg: sieving}, the elements in the current list `behave like' vectors distributed independently and uniformly at random over $\mathcal{S}_w^n$ in the sense that, even if some correlations appear in an iteration of the algorithm, it does not significantly affect the runtime of the algorithm.}}
\end{quote}

\noindent See \cite{DEEK23} for a more precise formulation of this heuristic and experimental verification.\footnote{A similar heuristic is used in the most efficient lattice-based sieving approaches (in particular, \cite{NV08,MV10} and their classical and quantum derivatives \cite{BDGL16,Laa16,AC:ChaLoy21,BCSS23}).}

\subsection{NNS using Locality-Sensitive Filtering (LSF)}\label{sec: NNS using LSF}

The algorithms for the NNS problem in the Hamming metric proposed in \cite{Car20, EPRINT:GuoJohNgu23, DEEK23} can all be formulated using the locality-sensitive filtering (LSF) framework in the Hamming metric. Notably, the best-performing algorithms for NNS in the Euclidean metric \cite{BDGL16} also employ LSF techniques.
We start by recalling this approach.

The underlying idea of locality-sensitive filtering (in $\F_2^n$) is the following: vectors $\xv, \yv$ satisfying $|\xv + \yv| = w$, also known as \textit{near neighbors}, can be found more efficiently if we restrict our search to \textit{local regions} of $\mathbb F_2^n$. The algorithm proceeds as follows. It starts with covering the space $\mathbb F_2^n$ with (potentially overlapping) regions, each region corresponding to a certain center $\mathbf{c} \in  \mathcal{C}_f$, for some $\mathcal{C}_f \subseteq \mathbb F_2^n$. Each vector $\mathbf{x}$ from the input list $\mathcal{L}$ is then \textit{filtered} according to the centers, namely, it is inserted into a \textit{bucket} corresponding to a center $\mathbf{c}$ if and only if $|\mathbf{x} \land \mathbf{c}| = \alpha$ for a fixed parameter $\alpha$.\footnote{The analog in the lattice setting is that a unit vector is added to a bucket corresponding to some unit vector $\cv \in \R^n$ if it has (absolute) inner product at least $\alpha$ with $\cv$.} 
This is what we refer to as the \textit{bucketing phase}. Each bucket thus contains all the elements from $\mathcal{L}$ that belong to the same region. The algorithm then searches for near neighbors within each bucket, which we refer to as the \textit{checking phase}. As the checking phase significantly affects the overall cost of the algorithm, our primary goal will be to improve its runtime by quantizing this part of the algorithm. 

We formalize the notion of bucket as follows, and define the set of \textit{valid centers} for a given vector. 

\begin{definition}[Bucket] 
    For $\cv \in \F_2^n$, integers $0 \leq \alpha \leq w \leq n$ with $\alpha \leq |\cv|$ and $\mathcal{L} \subseteq \Ss_w^n$, define
    $$\mathcal{B}_{\alpha}(\cv) := \{\xv \in \mathcal{L} \colon |\xv \land \cv| = \alpha\}.$$ 
\end{definition}

\begin{definition}[Valid centers]
    For $\xv \in \F_2^n$, integers $0 \leq \alpha \leq n$ with $\alpha \leq |\xv|$ and $\mathcal{C}_f \subseteq \F_2^n$, define
    \[\mathcal{VC}_{\alpha}(\xv) := \{\cv \in \mathcal{C}_f \colon |\xv \land \cv| = \alpha\}.\]
\end{definition}

The resulting NNS algorithm is given in Algorithm~\ref{alg: NNS}. It uses two subroutines: $\subr{FindValidCenters}$ and $\subr{FindSolutions}$. 
For a given vector $\mathbf{x} \in \mathcal{L}$, the subroutine $\subr{FindValidCenters}$ returns the set $\mathcal{VC}_\alpha(\mathbf{x}) = \{\mathbf{c} \in \mathcal{C}_f \ | \ |\mathbf{x} \land \mathbf{c}| = \alpha\}$. 
For a suitable choice of $\mathcal{C}_f$ (using so-called random product codes), the time complexity of this subroutine has been shown by \cite{DEEK23} to be $|\mathcal{VC}_\alpha(\xv)| + 2^{o(n)}$, as explained in the next section.

The subroutine $\subr{FindSolutions}$ returns all pairs $(\mathbf{x}, \mathbf{y})$ of vectors in the bucket $\mathcal{B}_{\alpha}(\mathbf{c})$ satisfying $|\mathbf{x} + \mathbf{y}| = w$. 
Classically, we obtain an algorithm for $\subr{FindSolutions}$ with expected runtime $\mathbb{E}[|\mathcal{B}_{\alpha}(\mathbf{c})|^2]$ by searching through all pairs in $\mathcal{B}_{\alpha}(\mathbf{c})$.  
(Although \cite{DEEK23} consider a slightly different approach for the checking phase, it results in the same asymptotic runtime for NNS.) 
In Section~\ref{sec: quantum speed-ups}, we focus on speeding up the search in the buckets by presenting  several quantum algorithms for $\subr{FindSolutions}$. 

\begin{algorithm}[h!]
    \DontPrintSemicolon
    \SetKwInOut{Input}{Input}\SetKwInOut{Output}{Output}
    
    \Input{weight $w$, input list $\mathcal{L} \subseteq \mathcal{S}_w^n$ of size $N$, set of centers $\mathcal{C}_f$, bucketing parameter $\alpha$} 
    \Output{list $\mathcal{L}'$ containing pairs $\mathbf{x}, \mathbf{y} \in \mathcal{L}$ with $|\mathbf{x} + \mathbf{y}| = w$}
    \vspace{0.2cm}
    \nonl \textsc{Initialization:}\\
    \For{$\mathbf{c} \in \mathcal{C}_f$}
    {
        $\mathcal{B}_{\alpha}(\mathbf{c}) = \emptyset$\;
    }
    \vspace{0.2cm}
    \nonl \textsc{Bucketing Phase:}\;    
    \For{$\mathbf{x} \in \mathcal{L}$}{
        \For{$\mathbf{c} \in \subr{FindValidCenters}(\mathcal{C}_f, \mathbf{x}, \alpha)$
        }{
        Add $\mathbf{x}$ to $\mathcal{B}_{\alpha}(\mathbf{c})$
        }
    }
    \vspace{0.2cm}
    \nonl \textsc{Checking Phase:}\;
    $\mathcal{L}' = \emptyset$\;
    \For{$\mathbf{c} \in \mathcal{C}_f$}{
        Add $\subr{FindSolutions}(\mathcal{B}_{\alpha}(\mathbf{c}))$ to $\mathcal{L}'$
    }
    \Return $\mathcal{L}'$
    \caption{$\NNS$ using locality-sensitive filtering}\label{alg: NNS}
\end{algorithm}

Depending on the size of $|\mathcal{C}_f|$, Algorithm~\ref{alg: NNS} has to be repeated a certain number of times to find all solutions. 
The size of $|\mathcal{C}_f|$ also affects other factors of the time and memory complexity of Algorithm~\ref{alg: NNS}, and should thus be chosen carefully. 
The best time and memory complexity in \cite[Cor.~4.2]{DEEK23} are obtained by choosing $|\mathcal{C}_f|$ such that $\mathbb{E}[|\mathcal{VC}_\alpha(\mathbf{x})|] = 2^{o(n)}$, ensuring that the expected runtime of $\subr{FindValidCenters}$ is $2^{o(n)}$.  
(An algorithm with similar time and memory complexities as \cite[Cor.~4.2]{DEEK23} was independently obtained by Carrier \cite[Cor.~8.2.6]{Car20}.)

More generally, given a classical or quantum algorithm for $\subr{FindSolutions}$, we obtain the following time and memory complexities for $\NNS(n,w,N)$, and thus for $\DP(n,k,w,N)$ by instantiating the oracle $\mathcal{A}_{\NNS}$ in Algorithm~\ref{alg: sieving} with the obtained algorithm for $\NNS(n,w,N)$. 

\begin{theorem}[Variant of {\cite[Cor.~8.2.6]{Car20}} and {\cite[Cor.~4.2, Theorem~3.2]{DEEK23}}]\label{thm: from NNS to DP}
    Consider $n, w, N \in \N$ with $w = \Theta(n)$ such that $\NNS(n,w,N)$ (Problem~\ref{prob: NNS}) is well-defined. For positive integers $v, \alpha = \Theta(n)$ and arbitrary $\xv,\yv \in \mathcal{S}_w^n$ such that $|\xv+\yv| = w$, let $|\mathcal{C}_f| = 2^{o(n)} \cdot |\mathcal{S}_v^n| / |\mathcal{C}_{v,\alpha}^n(\xv)|$ and $R = 2^{o(n)} \cdot |\mathcal{C}_{v,\alpha}^n(\xv)| / |\mathcal{W}_{v,\alpha}^n(\xv, \yv)|$. 
    
    Then, given a classical, resp.\ quantum, algorithm for $\subr{FindSolutions}$ with expected runtime $T$, there exists a classical, resp.\ quantum, algorithm that solves $\NNS(n,w,N)$ in expected time $$R \cdot (2^{o(n)}\cdot N + |\mathcal{C}_f| \cdot T)$$
    using $R$ calls to Algorithm~\ref{alg: NNS}.
    Moreover, the classical memory is of expected size $2^{o(n)} \cdot N$, and the other memory complexities are the same as for the $\subr{FindSolutions}$ subroutine. 

    In addition, if $4 \cdot {n \choose w}\big/\big({w \choose w/2} {n-w \choose w/2}\big) \leq N \leq {n \choose w}/{2^{n-k}}$ and $k = \Theta(n)$ is such that the binary-sieve heuristic in \cite{DEEK23} holds for $(n,k,w,N)$, then there exists a classical, resp.\ quantum, sieving-based algorithm that solves $\DP(n,k,w,N)$ (Problem~\ref{prob: decoding problem, N sol}) with the same time and memory complexity, up to polylogarithmic factors. 
\end{theorem}

Note that the first part follows from {\cite[Cor.~8.2.6]{Car20}} and {\cite[Cor.~4.2]{DEEK23}} when allowing for a quantum subroutine, and that the second part follows from \cite[Theorem~3.2]{DEEK23}.

\subsection{Random Product Codes and an Efficient Algorithm for $\subr{FindValidCenters}$}\label{sec: RPC}

To efficiently perform the $\subr{FindValidCenters}$ subroutine, \cite{DEEK23} suggest to let $\mathcal{C}_f$ be a binary linear code for which there exists an efficient decoding algorithm. More precisely, they use the notion of \textit{random product codes}, originally introduced in \cite{BDGL16} for $\R^n$, and defined as follows for $\F_2^n$.   

\begin{remark}
    A similar notion of random product codes and a corresponding decoding algorithm was also independently presented (in French) in \cite[Section~9.1]{Car20} for a related problem.  
\end{remark}

\begin{definition}[Random product code (RPC) in $\F_2^n$]
    An $(n,v,t)$-RPC of size $\kappa$ is an element $\mathcal{C}$ drawn uniformly at random from the set \begin{align*}
        R_{n,v,t, \kappa} := \{\mathcal{C} = \mathcal{C}^{(1)} \times \cdots \times \mathcal{C}^{(t)} \colon \mathcal{C}^{(i)} \subseteq S_{v/t}^{n/t} \text{ such that } \forall i, |\mathcal{C}^{(i)}| = \kappa^{1/t}\}.
    \end{align*}
    We write $\mathcal{C} \sim R_{n,v,t, \kappa}$. 
\end{definition}

\noindent The set of centers is then sampled uniformly at random from $R_{n,v,t,\kappa}$ to guarantee:

\begin{outline}
    \1[(1)] Efficient decodability: in some reasonable parameter regimes, one can compute $\mathcal{VC}_\alpha(\xv) := \{\cv \in \mathcal{C} \colon |\xv \land \cv| = \alpha\}$ in time asymptotically equal to its size $|\mathcal{VC}_\alpha(\xv)|$. 
    \1[(2)] Random behavior: for $t$ not too large, a sample $\mathcal{C}$ behaves like a random code in the sense that the success probability that $\mathcal{C} \sim R_{n,v,t, \kappa}$ `captures' a pair $(\xv,\yv)\in \mathcal{L}^2$ is the same (up to factors subexponential in $n$) as for a random code in $S_v^n$.\footnote{This `randomness' property puts a constraint on $t$. More precisely, as described in \cite{DEEK23} (and in \cite{BDGL16} for $\R^n$), we would like $t$ to be small enough to guarantee that we can approximate a cap/wedge in $\F_2^n$ by the Cartesian product of $t$ caps/wedges in $\F_2^{n/t}$. By Lemma~4.6 and Lemma~4.7 in \cite{DEEK23}, these approximations are satisfied up to a subexponential factor in $n$ if $t = o(n/\log(n))$.} 
\end{outline}

\noindent For more details, we refer to \cite{DEEK23}. As we will refer back to the first property later in this work, we state the corresponding result here (which implicitly follows from the application of \cite[Lemma~4.5]{DEEK23} in the proof of \cite[Theorem~4.4]{DEEK23}).

\begin{theorem}[Efficient decodability of RPC \cite{DEEK23}]\label{thm: properties RPC}
Let $n \in \mathbb{N}$ and let $w, v, \alpha = \Theta(n), \kappa$ be positive integers. For $t = \Theta(\sqrt{n})$, let $\mathcal{C}$ be an $(n,v,t)$-RPC of size exponential in $n$. 
Then there exists a classical algorithm that, for any $\xv \in \Ss_w^n$, computes the set $\mathcal{VC}_\alpha(\xv) := \{\cv \in \mathcal{C} \colon |\xv \land \cv| = \alpha\}$ in time $|\mathcal{VC}_\alpha(\xv)| + 2^{o(n)}$. 
\end{theorem}

\section{Quantum Algorithms for NNS and Code Sieving}
\label{sec: quantum speed-ups} 
\label{sec: quantum NNS}
In this section, we present the first quantum algorithms for code sieving.
More precisely, we describe different quantum algorithms for the \subr{FindSolutions} procedure in the NNS subroutine (Algorithm~\ref{alg: NNS}), and analyze the time and memory complexities of our algorithms. 
By plugging our results (i.e., \Cref{thm: bucket search using Grover} and \Cref{thm: bucket search using QW and RPC}) into \Cref{thm: from NNS to DP}, we obtain the time and memory complexity of the resulting quantum algorithms for NNS and, hence, code sieving. 
Our numerical results in Section~\ref{sec: numerical results} illustrate the quantum speed-ups obtained for these resulting algorithms.

We recall the context in which \subr{FindSolutions} is used. That is, consider the successful completion of the bucketing phase of Algorithm~\ref{alg: NNS} for an input list $\mathcal{L} \subseteq \Ss_w^n$: for some bucketing parameter $\alpha$, we have sampled a set $\mathcal{C}_f \subseteq \Ss_v^n$ of bucket centers and created a data structure containing the buckets
\begin{equation} \label{eq:bucket}
    \mathcal{B}_\alpha(\cv) := \{\xv \in \mathcal{L} \colon |\xv \land \cv| = \alpha\} \subseteq \mathrm{Region}_\alpha(\cv)  := \{\xv \in \Ss_w^n : |\xv \wedge \cv|=\alpha\}
\end{equation}
for all $\cv \in \mathcal{C}_f$.
Then the goal of \texttt{FindSolutions} is to find all $\xv,\yv \in \mathcal{B}_\alpha(\cv)$ satisfying $|\xv + \yv| = w$.

Note that if $\mathcal{L} \subseteq \Ss_w^n$ is sampled independently and uniformly at random, then the vectors in $\mathcal{B}_\alpha(\cv)$ are distributed independently and uniformly over $\mathrm{Region}_\alpha(\cv)$.  
Hence, formally, our quantum algorithms for \texttt{FindSolutions} solve the following problem.

\begin{problem}[Bucket search]\label{prob: bucket search} 
    Let $B, n, w, v, \alpha \in \N$ with $w, v \leq n$ and $\alpha \leq \max\{v,w\}$. Let $\cv \in \Ss_v^n$. 
    Given a list $\mathcal{B}_\alpha(\cv)$ of $B$ vectors that are sampled independently and uniformly at random from $\{\xv \in \Ss_w^n : |\xv \wedge \cv|=\alpha\}$, find a $(1-o(1))$-fraction of all pairs $(\xv, \yv) \in \mathcal{B}_\alpha(\cv)^2$ satisfying $|\xv+\yv| = w$ (called \textit{solution pairs}). 
    We say that this problem has \textit{parameters} $(B, n, w, v, \alpha)$. 
\end{problem}

This problem is a variant of the near-neighbor search problem (Problem~\ref{prob: NNS}) defined in Section~\ref{sec: sieving for codes}.
An important difference here is that the input list $\mathcal{B}_\alpha(\cv)$ is not uniformly random on $\Ss_w^n$ but on the set $\{\xv \in \Ss_w^n : |\xv \wedge \cv|=\alpha\}$.  
Consequently, the probability (denoted $p$ below) that a uniformly random pair from the input list forms a solution pair does not only depend on $n,w$, but also on the parameters $v, \alpha$, as further shown in the next section.

\begin{remark}[Expected number of solutions]
    Note that the expected number of solutions to Problem~\ref{prob: bucket search} with parameters $(B, n, w, v, \alpha)$ is ${B \choose 2} p$, where $p$ denotes the probability that a uniformly random pair $\xv, \yv$ in $\mathcal{B}_\alpha(\cv)$ forms a solution pair.
\end{remark}

In the remainder of Section~\ref{sec: quantum speed-ups}, we fix parameters $(B,n,w,v,\alpha)$. 
Before presenting our quantum algorithms for solving Problem~\ref{prob: bucket search}, we calculate the probability $p$ of forming a solution pair, as it is used in the analysis of our algorithms.

\subsection{Probability $p$ of Forming a Solution Pair} \label{sec:p} 

As before, let $p$ denote the probability that a uniformly random pair $\xv, \yv \in \mathcal{B}_\alpha(\cv)$ satisfies $|\xv+\yv|=w$. 
That is,  
\begin{align} 
    p &:= \Pr_{\xv,\yv\in \mathcal{B}_\alpha(\cv)}[|\xv + \yv| = w] = \Pr_{\xv,\yv\in \Ss_w^n}[|\xv + \yv| = w \mid |\xv \land \cv| = \alpha \text{ and } |\yv \land \cv| = \alpha]. \label{eq: p}
\end{align}
We can express $p$ in terms of the parameters $(n, w, v, \alpha)$ as follows. 

\begin{lemma}[Probability $p$]\label{lem: probability p (wedge version)}
    The probability $p$ as defined in \Cref{eq: p} is equal to \begin{align*}
    p(n,w,v,\alpha) = \frac{{w \choose w/2} {n - w \choose w/2}}{ {v \choose \alpha} {n - v \choose w - \alpha} {w \choose \alpha} {n - w \choose v - \alpha}} \cdot |\mathcal{W}_{v,\alpha}^n(\xv,\yv)|
\end{align*} 
where $|\mathcal{W}_{v,\alpha}^n(\xv,\yv)| = \sum_{e = \max(0, 2\alpha - v)}^{\min(\alpha, w/2)} {w/2 \choose e} {w/2 \choose \alpha - e}^2 {n - 3w/2 \choose v - 2\alpha + e}$ is the surface area of a wedge for arbitrary $\xv, \yv \in \Ss_w^n$ with $|\xv + \yv| = w$.\footnote{Recall (see Remark~\ref{rem: surface area only depends on weight}) that $\mathcal{W}_{v,\alpha}^n(\xv,\yv)$ does not depend on the specific choice of $\xv$ and $\yv$, but only on $|\xv|, |\yv|, |\xv + \yv|$. Consequently, the probability $p$ is independent of the particular selection of $\xv$ and $\yv$.} 
\end{lemma}


Before proceeding with the proof of \Cref{{lem: probability p (wedge version)}}, we state the following lemma.\footnote{In the following proofs, we will use the fact that for all $a,b,c,d \in \N$ (for which the binomial coefficients are well-defined),  ${a \choose b} {b \choose d} {a-b \choose c-d} = {a \choose c} {c \choose d} {a-c \choose b-d}$.} 
Also, it might help to keep Figure~\ref{fig: overlap x,y,c} in mind. 

\begin{figure}[h!]
    \centering
    \begin{tikzpicture}[scale=0.5]
    \fill[blue!20] (0,4)--(4,4)--(4,5)--(0,5)--cycle ; 
    \draw[decorate,decoration={brace,amplitude=3pt,mirror,raise=4ex}] (0,5)--(4,5) node[midway,yshift=-2.5em]{$w$};
    
    \fill[blue!20] (2,0)--(6,0)--(6,1)--(2,1)--cycle ; 
    \draw[decorate,decoration={brace,amplitude=3pt,mirror,raise=4ex}] (2,1)--(4,1) node[midway,yshift=-2.5em]{$w/2$};
    \draw[decorate,decoration={brace,amplitude=3pt,mirror,raise=4ex}] (4,1)--(6,1) node[midway,yshift=-2.5em]{$w/2$}; 

    \fill[blue!20] (0,2)--(0.6,2)--(0.6,3)--(0,3)--cycle ; 
    \draw[decorate,decoration={brace,amplitude=3pt,mirror,raise=4ex}] (0,2+1)--(0.6,2+1) node[midway,yshift=-2.5em]{$\alpha-e$}; 
    
    \fill[blue!20] (3,2)--(4,2)--(4,3)--(3,3)--cycle ; 
    \draw[decorate,decoration={brace,amplitude=3pt,mirror,raise=4ex}] (3,2+1)--(4,2+1) node[midway,yshift=-2.5em]{$e$}; 
    
    \fill[blue!20] (5.4,2)--(6,2)--(6,3)--(5.4,3)--cycle ; 
    \draw[decorate,decoration={brace,amplitude=3pt,mirror,raise=4ex}] (5.4,2+1)--(6,2+1) node[midway,yshift=-2.5em]{$\alpha-e$}; 

    \fill[blue!20] (7.7,2)--(10.5,2)--(10.5,3)--(7.7,3)--cycle ; 
    \draw[decorate,decoration={brace,amplitude=3pt,mirror,raise=4ex}] (7.7,2+1)--(10.5,2+1) node[midway,yshift=-2.5em]{$v-2\alpha+e$};

    \draw (0,4)--(10.5,4)--(10.5,5)--(0,5)--cycle;
    \draw (0,2)--(10.5,2)--(10.5,3)--(0,3)--cycle;
    \draw (0,0)--(10.5,0)--(10.5,1)--(0,1)--cycle;
    \draw (-0.5,4.5) node[left] {$\textbf{x}=$};
    \draw (-0.5,2.5) node[left] {$\textbf{c}=$};
    \draw (-0.5,0.5) node[left] {$\textbf{y}=$};

    \draw[black!40,dashed] (0.6,4)--(0.6,5);
    \draw[black!40,dashed] (0.6,0)--(0.6,1);
    \draw[black!40,dashed] (2,2)--(2,3);\draw[black!40,dashed] (2,4)--(2,5);\draw[black!40,dashed] (3,4)--(3,5);
    \draw[black!40,dashed] (3,0)--(3,1);
    \draw[black!40,dashed] (4,0)--(4,1);
    \draw[black!40,dashed] (6,4)--(6,5);
    \draw[black!40,dashed] (5.4,4)--(5.4,5);
    \draw[black!40,dashed] (5.4,0)--(5.4,1);
    \draw[black!40,dashed] (7.7,0)--(7.7,1);
    \draw[black!40,dashed] (7.7,4)--(7.7,5);
\end{tikzpicture}
    \caption{An example of the overlaps between $\xv,\yv \in \Ss_w^n$ and $\cv \in \Ss_v^n$, where $|\xv + \yv| = w$, $|\xv \land \cv| = |\yv \land \cv| = \alpha$, and $|\xv \land \yv \land \cv| = e$. All cases of such vectors under these conditions are permutations of this example.}
    \label{fig: overlap x,y,c}
\end{figure}

\begin{lemma}\label{lem: probability random solution is in fixed bucket}
    For all $\cv \in \Ss_v^n$ and all $\xv',\yv' \in S_w^n$ satisfying $|\xv' + \yv'| = w$, we have  \begin{align*}
        \Pr_{\xv,\yv \in_R \Ss_w^n}[|\xv \land \cv| = \alpha \text{ and } |\yv \land \cv| = \alpha \ \big| \ |\xv \land \yv| = w/2] = \frac{|\mathcal{W}_{v,\alpha}^n(\xv',\yv')|}{|\Ss_v^n|}.
    \end{align*}
\end{lemma}

\begin{proof}
For fixed $\cv \in \Ss_v^n$ and $\xv',\yv' \in S_w^n$ satisfying $|\xv' + \yv'| = w$, the stated probability is equal to \begin{align*}
        &\, \frac{|\{(\xv,\yv) \in (\Ss_w^n)^2 \colon |\xv \land \cv| = \alpha, |\yv \land \cv| = \alpha, |\xv \land \yv| = w/2\}|}{|\{(\xv,\yv) \in (\Ss_w^n)^2 \colon |\xv \land \yv| = w/2\}|}  \\
        =&\, \frac{\sum_e {v \choose \alpha} {n-v \choose w-\alpha} {\alpha \choose e} {v-\alpha \choose \alpha -e} {w-\alpha \choose w/2 -e} {n - (w + v - \alpha) \choose w/2 - \alpha + e}}{{n \choose w} {w \choose w/2} {n-w \choose w/2}}
\end{align*}
where $e$ ranges over all possible values of $|\xv \land \yv \land \cv|$. Now, using the equalities $\frac{{v \choose \alpha} {n-v \choose w-\alpha}}{{n \choose w}} = \frac{{w \choose \alpha} {n-w \choose v-\alpha}}{{n \choose v}}$, ${w \choose \alpha} {\alpha \choose e} {w - \alpha \choose w/2 - e} = {w \choose w/2} {w/2 \choose e} {w - w/2 \choose \alpha - e}$, and ${n - w \choose v - \alpha} {v- \alpha \choose \alpha - e} {(n - w) - (v - \alpha) \choose w/2 - (\alpha - e)} = {n - w \choose w/2} {w/2 \choose \alpha - e} {(n - w) - w/2 \choose (v - \alpha) - (\alpha -  e)}$, it can be seen that this indeed equals $\frac{|\mathcal{W}_{v,\alpha}^n(\xv',\yv')|}{{n \choose v}}$. 
\qed
\end{proof}

\begin{proof}[Proof of Lemma~\ref{lem: probability p (wedge version)}] 
By definition, \begin{align*}
    p &= \Pr_{\xv,\yv\in_R \Ss_w^n}[|\xv \land \yv| = w/2 \mid |\xv \land \cv| =  \alpha \text{ and } |\yv \land \cv| = \alpha] \\
    &= \Pr_{\xv,\yv\in_R \Ss_w^n}[ |\xv \land \cv| = \alpha \text{ and } |\yv \land \cv| = \alpha \mid |\xv \land \yv| = w/2 ] \cdot \frac{\Pr_{\xv,\yv\in_R \Ss_w^n}[ |\xv \land \yv| = w/2 ]}{\Pr_{\xv,\yv\in_R \Ss_w^n}[ |\xv \land \cv| = |\yv \land \cv| = \alpha]} \\
    &= \frac{|\mathcal{W}_{v,\alpha}^n(\xv',\yv')|}{{n \choose v}} \cdot \frac{\Pr_{\xv,\yv\in_R \Ss_w^n}[ |\xv \land \yv| = w/2 ]}{\Pr_{\xv,\yv\in_R \Ss_w^n}[ |\xv \land \cv| = |\yv \land \cv| = \alpha]}
\end{align*}
for arbitrary $\xv',\yv' \in S_w^n$ satisfying $|\xv' + \yv'| = w$. Note that the last equality follows from Lemma~\ref{lem: probability random solution is in fixed bucket}. 
Since $\Pr_{\xv,\yv\in_R \Ss_w^n}[ |\xv \land \yv| = w/2 ] = \frac{{n \choose w} {w \choose w/2} {n-w \choose w/2}}{{n \choose w}^2}$ and  $\Pr_{\xv,\yv\in_R \Ss_w^n}[ |\xv \land \cv| = \alpha \text{ and } |\yv \land \cv| = \alpha] = \frac{{v \choose \alpha}^2 {n-v \choose w-\alpha}^2}{{n \choose w}^2}$, it follows that \begin{align*}
    p = \frac{|\mathcal{W}_{v,\alpha}^n(\xv',\yv')|}{{n \choose v}} \frac{{n \choose w} {w \choose w/2} {n-w \choose w/2}}{{v \choose \alpha}^2 {n-v \choose w-\alpha}^2}. 
\end{align*}
Notice that ${n \choose v} {v \choose \alpha} {n-v \choose w-\alpha} = {n \choose w} {w \choose \alpha} {n-w \choose v-\alpha}$, hence the desired result follows. 
\qed
\end{proof}

\subsection{Quantum Algorithm for \subr{FindSolutions} Using Grover's Algorithm} \label{sec:grover}

As a warm-up, we give a straightforward quantum algorithm for solving Problem~\ref{prob: bucket search} (and thus \subr{FindSolutions}) using Grover's algorithm \cite{Gro96}. 
This will be used as a baseline to compare our more advanced algorithms with. 

\begin{theorem}[\subr{FindSolutions} using Grover]\label{thm: bucket search using Grover}
    Let $(B, n, w, v, \alpha)$ be well-defined parameters for Problem~\ref{prob: bucket search} and let $p = p(n,w,v,\alpha)$ be according to \Cref{lem: probability p (wedge version)}. 
    Then there exists a quantum algorithm that solves Problem~\ref{prob: bucket search} with parameters $(B, n, w, v, \alpha)$ in expected time $\Tilde{O}(B^2 \sqrt{p})$ using classical memory and QRACM of expected size $M_C = M_{QRACM} = B$ and quantum memory of expected size $M_Q = n^{O(1)}$. 
\end{theorem}

\begin{proof}
    The statement follows from repeatedly applying Grover's algorithm (Theorem~\ref{thm: Grover}) with the set of all ${B \choose 2}$ pairs as the search space. The expected number of solutions is $t:= {B \choose 2} p$. Using a coupon-collector argument, it follows that all solutions can be found in expected time $\Tilde{O}(\sqrt{t {B \choose 2}} ) = \Tilde{O}(B^2 \sqrt{p})$.  
    The $B$ list vectors are classically stored, and Grover's algorithm uses quantum access to them (QRACM); it also uses a polynomial (in $n$) number of qubits.
\qed
\end{proof}

This Grover-based quantum algorithm can be used as the subroutine \subr{FindSolutions} in \Cref{thm: from NNS to DP} to obtain a quantum algorithm for NNS, and thus for code sieving. 

\subsection{Quantum Algorithm for \subr{FindSolutions} Using Quantum Walks} \label{sec:quwalk}

We will now replace Grover's algorithm with quantum-walk techniques. In particular, we describe a $\subr{FindSolutions}$ subroutine that searches for solution pairs ($\xv, \yv$ such that $|\xv + \yv| = w$) inside a given list $\mathcal{B}_{\alpha}(\mathbf{c})$ using a quantum walk. 
To speed up the search for solution pairs, we add a layer of locality-sensitive filtering (LSF) as part of the data of each vertex. 
This idea was inspired by \cite{AC:ChaLoy21}, and we show that we can indeed transfer most of their ideas from the Euclidean metric to the Hamming metric.
We analyze the complexity of the resulting quantum walk, and consider a version of the walk where we apply the `sparsification' technique from \cite{AC:ChaLoy21}, resulting in a different balance of the complexity parameters.

\paragraph{MNRS-style quantum walk.} 
We use the quantum-walk framework from \cite{MNRS11} (see \Cref{sec: quantum prelims}). Our quantum walk is defined on the Johnson graph $J(B, s)$, for some parameter $s \leq B$ that must be carefully chosen later. The vertices of the graph are identified with the size-$s$ subsets of the input list $\mathcal{B}_\alpha(\cv)$, and two vertices are adjacent if and only if the corresponding subsets differ in exactly one element. 
The goal of the quantum walk is to return a \textit{marked vertex}, which is a vertex containing a solution pair. As aforementioned, we add some additional \textit{data} to each vertex to speed up the search for marked vertices. 

We briefly sketch the inner steps of the walk (for more details see \Cref{sec: proof of bucket search using QW}). 
The set-up step (of cost $\mathsf{S}$) constructs a uniform superposition over all vertices and their corresponding data. 
In the update step (of cost $\mathsf{U}$), the current vertex (i.e., a size-$s$ subset of vectors) is mapped to a uniform superposition of its neighbors. During the update step, the algorithm also keeps track of whether it has encountered a solution pair, which will be facilitated by the data associated to the vertices. As a result, the checking step is immediate (and thus has cost $\mathsf{C} = \Tilde{O}(1)$). 
Writing $\epsilon$ for the fraction of marked vertices and $\delta$ for the spectral gap of the graph, the quantum walk returns a marked vertex in time $\mathsf{S}+\frac{1}{\sqrt{\epsilon}}(\frac{1}{\sqrt{\delta}}\mathsf{U}+\mathsf{C})$. 
Each run of our quantum walk finds one solution to Problem~\ref{prob: bucket search}, so we repeat it until a $(1-o(1))$-fraction of all solution pairs is found.\footnote{The reader might notice we do not apply the reusable quantum-walk techniques from \cite{BCSS23}. The reason is that our numerical experiments do not show a significant speed-up, see Remark~\ref{rem: reusable walk}.}

\paragraph{Adding a layer of LSF.}  
We will now describe the data added to each vertex of the Johnson graph $J(B, s)$ to facilitate the detection of marked vertices during the update step. The first part of the update step can essentially be viewed as a map between two adjacent vertices (which in $J(B,s)$ differ by exactly one element): given a subset $S \subseteq \mathcal{B}_\alpha(\cv)$, a neighbor $S'$ of $S$ is obtained by removing a vector $\xv_{old}$ from $S$ and adding a vector $\xv_{new}$ from $\mathcal{B}_\alpha(\cv) \setminus S$. 
In the second part of the update step, the algorithm checks whether the newly added vector $\xv_{new}$ forms a solution pair with one of the non-removed vertices (i.e., those in $S \setminus \{\xv_{old}\}$). This check can be performed by simply applying Grover, but we achieve better complexities if we add a layer of locality-sensitive filtering.

More precisely, we identify the input list $\mathcal{B}_{\alpha}(\cv) \subseteq S_w^n$ with the list $\mathcal{L}' = \{\pi_{\cv}(\xv)  \colon \xv \in \mathcal{B}_{\alpha}(\cv)\} \subseteq \Ss_{\alpha}^v$, where $\pi_{\cv}(\xv) := \xv \land \cv$ is the $v$-dimensional vector obtained by projecting $\xv$ onto the support of $\cv$.\footnote{Note that $\xv \land \cv$ is actually an $n$-dimensional vector. With abuse of notation, we sometimes view it as a $v$-dimensional vector. The dimension should be clear from the context.} As further detailed in \Cref{sec: reduction to a variant of NNS}, this allows us to apply the LSF techniques from \Cref{sec: NNS using LSF} and \Cref{sec: RPC}. 
In particular, at the start of our algorithm, we sample an RPC $\mathcal{C}_f' \subseteq S_{v'}^v$ for some parameter $v'$. Each vector $\cv' \in \mathcal{C}_f'$ forms the center of a bucket $\mathcal{B}_{\beta}(\mathbf{c}')$ for some bucketing parameter $\beta$. However, these buckets will only be defined per vertex: for a vertex corresponding to subset $S$, we fill $\mathcal{B}_{\beta}(\mathbf{c}')$ with the vectors in $S$ that are in a certain sense `close' to $\cv'$ with respect to the parameter $\beta$. Specifically, we let 
    \begin{equation*} 
        \mathcal{B}_{\beta}(\mathbf{c}') := \{\xv \in S \colon |\pi_{\cv}(\xv) \land \cv'| = \beta\} 
    \end{equation*}
The filled buckets $(\mathcal{B}_{\beta}(\mathbf{c}'))_{\cv' \in \Cs_f'}$ are then added as part of the data of the vertex corresponding to $S$.

Consequently, instead of checking for solution pairs within a vertex $S$, the algorithm will only check for solution pairs within the $|\mathcal{C}_f'|$ $\beta$-buckets. This may induce some false negatives, but this can be controlled by choosing $\beta$ carefully. 
Just like the parameter $\alpha$ was used to quantify the probability that two vectors in an $\alpha$-bucket $\mathcal{B}_{\alpha}(\cv)$ form a solution to Problem~\ref{prob: NNS} (NNS), the parameter $\beta$ quantifies the probability that two vectors in a $\beta$-bucket $\mathcal{B}_{\beta}(\cv')$ form a solution to Problem~\ref{prob: bucket search} (bucket search).
By choosing the parameters (including $\beta$) carefully, it allows us to obtain speed-ups, which is further demonstrated by our numerical results in \Cref{sec: numerical results}.

\

To state our main result on our quantum-walk algorithm, we need the following notions. 
For fixed $t' \in \Theta(\sqrt{v})$ and $C$, we define $q = q_C = \Pr[\exists \cv' \in \mathcal{C}_f' \text{ s.t. } |\pi_{\cv}(\xv) \land \cv'| = \beta \text{ and } |\pi_{\cv}(\yv) \land \cv'| = \beta]$, where the probability is taken over $\mathcal{C}_f' \sim R_{v,v',t', C}$ and uniformly random $\xv, \yv \in \mathrm{Region}_{\alpha}(\cv)$ satisfying $|\xv + \yv| = w$. Informally, for an arbitrary solution pair $(\xv, \yv)$, $q$ denotes the probability that they share a valid $\beta$-center. 
Furthermore, for any $\mathcal{C}'_f \sim R_{v, v', t', C}$, we define  
\begin{equation*} 
    \mathcal{VC}_\beta(\mathbf{x}) := \{ \cv' \in \mathcal{C}'_f \colon |\pi_{\cv}(\xv) \land \cv'| = \beta\}
\end{equation*} 
to be the set of valid $\beta$-bucket centers $\cv'$ for any given $\xv \in \mathrm{Region}_\alpha(\cv)$. 
We write $\mathbb{E}[|\mathcal{VC}_\beta(\mathbf{x})|]$ for the expected size of $\mathcal{VC}_\beta(\mathbf{x})$, where the expectation is taken over $\mathcal{C}_f' \sim R_{v,v',t', C}$. 
Furthermore, we write $\mathbb{E}[|\mathcal{B}_{\beta}(\mathbf{c}')|]$ for the expected size of a bucket $\mathcal{B}_{\beta}(\mathbf{c}')$, where the expectation is taken over $\xv_1, \ldots, \xv_s \in \mathrm{Region}_\alpha(\cv)$ sampled independently and uniformly at random. Note that $\mathbb{E}[|\mathcal{VC}_\beta(\mathbf{x})|]$ and $\mathbb{E}[|\mathcal{B}_{\beta}(\mathbf{c}')|]$ are the same for all $\xv \in \mathrm{Region}_\alpha(\cv)$ and $\cv' \in \Ss_{v'}^v$, respectively. 

The proof of \Cref{thm: bucket search using QW and RPC} will be given in \Cref{sec: proof of bucket search using QW}.  

\begin{remark}\label{rem: writing max(1, exp size) as exp size}
    With a slight abuse of notation, we write $\mathbb{E}[|\mathcal{VC}_{\beta}(\mathbf{x})|])$ for
    $\max(1, \mathbb{E}[|\mathcal{VC}_{\beta}(\mathbf{x})|])$ to simplify the expressions for time and memory complexity. We do the same for $\mathbb{E}[|\mathcal{B}_{\beta}(\mathbf{c}')|]$. 
\end{remark}

\begin{theorem}[\subr{FindSolutions} using quantum walks]\label{thm: bucket search using QW and RPC}
\label{thm:BucketSearchQW}
    Let $(B, n, w, v, \alpha)$ be well-defined parameters for Problem~\ref{prob: bucket search} and let $p = p(n,w,v,\alpha)$ be according to \Cref{lem: probability p (wedge version)}. 
    For non-negative integers $s = \Tilde{O}(\frac{1}{\sqrt{p}})$,  $\beta \leq \alpha$, $v' \leq v$, $t' = \Theta(\sqrt{v})$, and $C$, define $q = q_C$, $\mathbb{E}[|\mathcal{VC}_{\beta}(\mathbf{x})|]$, and $\mathbb{E}[|\mathcal{B}_{\beta}(\mathbf{c}')|]$ as above. 
    
    Then there exists a quantum-walk algorithm  that solves Problem~\ref{prob: bucket search} with parameters $(B, n, w, v, \alpha)$ in expected time  $\Tilde{O}\left( B^2 p \cdot  (\mathsf{S} + \frac{1}{\sqrt{\epsilon}}(\frac{1}{\sqrt{\delta}} \mathsf{U} + \mathsf{C})) \right)$,
    where, omitting factors subexponential in $n$,
    \begin{align*}
      \mathsf{S} &= s \cdot \mathbb{E}[|\mathcal{VC}_\beta(\mathbf{x})|] \\
      \mathsf{U} &= \mathbb{E}[|\mathcal{VC}_\beta(\mathbf{x})|] + \sqrt{\mathbb{E}[|\mathcal{VC}_\beta(\mathbf{x})|] \cdot \mathbb{E}[|\mathcal{B}_{\beta}(\mathbf{c}')|]} \\
      \mathsf{C} &= 1 \\
      \delta &= \frac{1}{s} \\
      \epsilon &= s^2 p q.
    \end{align*}
    
    This quantum-walk algorithm uses classical memory and QRACM of expected size $\Tilde{O}(B)$, and quantum memory and QRAQM of expected size $\Tilde{O}(s \cdot \mathbb{E}[|\mathcal{VC}_\beta(\mathbf{x})|])$.
\end{theorem}

Note that the expected size of $\mathcal{VC}_\beta(\mathbf{x})$ plays an important role in the complexity of our quantum-walk algorithm. In \Cref{sec: sparsification}, we describe how this expected size (and thus the complexity) crucially depends on the choice of $|\Cs_f'|$, and we analyze two different choices. 
Besides the choice of $|\Cs_f'|$,  the precise complexity of \Cref{thm: bucket search using QW and RPC} also depends on the optimal choice of the parameters $s$, $v'$, and $\beta$. 

Similar to the Grover-based algorithm, we can use the above quantum-walk algorithm as the subroutine \subr{FindSolutions} in \Cref{thm: from NNS to DP} to obtain a quantum algorithm for NNS and code sieving.

\subsection{Adding a Layer of Filtering}\label{sec: reduction to a variant of NNS} 

In this section, we present the technical details to justify adding a layer of LSF to the data of our quantum walk. Recall that our input list $\mathcal{B}_\alpha(\cv)$ is a subset of $\{\xv \in \Ss_w^n : |\xv \wedge \cv|=\alpha\}$. In order to apply the LSF techniques from \cite{DEEK23} (described in Section~\ref{sec: RPC}), we will relate the problem of finding solution pairs in $\mathcal{B}_\alpha(\cv)$ to a problem where we search for solution pairs in a list $\mathcal{L}'$ of vectors sampled independently and uniformly over a suitable (smaller-dimensional) Hamming sphere. 

This is motivated by the following observations. 
Note that any $\xv \in \mathcal{B}_{\alpha}(\cv)$ can be written as $\xv = \xv \land \cv + \xv \land \overline{\cv}$ for the center $\cv$. Suppose, for a moment, that $\cv$ is of weight $\alpha$ (i.e., $v = \alpha$). Then, for any $\xv,\yv \in \mathcal{B}_{\alpha}(\cv)$ we have that $\xv \land \cv = \cv$ and $\yv \land \cv = \cv$, and thus $\xv + \yv = \xv \land \overline{\cv} + \yv \land \overline{\cv}$. In other words, for $\xv,\yv \in \mathcal{B}_{\alpha}(\cv)$, we have that $|\xv + \yv| = w$ if and only if $|\xv \land \overline{\cv} + \yv \land \overline{\cv}| = w$. Therefore, we could restrict to searching for solutions among the vectors $\xv \land \overline{\cv}$ instead. 
However, if $v > \alpha$, then this equivalence breaks down: whether $\xv + \yv$ is of weight $w$ does no longer only depend on $|\xv \land \overline{\cv} + \yv \land \overline{\cv}|$, but also on $|\xv \land \yv \land \cv|$ (that is, the latter need no longer be $0$). Specifically, for a fixed pair $\xv,\yv \in \mathcal{B}_{\alpha}(\cv)$, we have that $|\xv + \yv| = w$ if and only if $|\xv \land \overline{\cv} + \yv \land \overline{\cv}| = w - 2\alpha + 2e$, where $e := |\xv \land \yv \land \cv|$. Although the value of $e$ may differ for different pairs $(\xv,\yv)$, Theorem~\ref{thm: complexity of quantum ISD assuming existence oracle} shows that for most solution pairs $(\xv, \yv)$ it is the same. The proof will be given at the end of this section. 

\begin{theorem}\label{thm: existence residuals for codes}
    Let $\mathcal{B}_\alpha(\cv)$ be an instance of Problem~\ref{prob: bucket search} with parameters $(B,n,w,v,\alpha)$. There exists an integer $e^*$ such that the expected number of pairs $(\xv,\yv) \in \mathcal{B}_{\alpha}(\cv)^2$ satisfying $|\xv + \yv| = w$ is equal, up to a multiplicative factor that is linear in $n$, to the expected number of pairs $(\xv,\yv)\in\mathcal{B}_{\alpha}(\cv)^2$ satisfying $|\xv + \yv| = w$ and $|\xv \land \yv \land \cv| = e^*$. 
\end{theorem}

\noindent In other words, Problem~\ref{prob: bucket search} is essentially equivalent to the following problem:
\begin{quote}
    {\normalsize \textit{Find (almost) all $\xv,\yv$ in $\mathcal{B}_{\alpha}(\cv)$ satisfying $|\xv + \yv| = w$ and $|\xv \land \yv \land \cv| = e^*$.}}
\end{quote}
We will therefore restrict to solving this latter problem instead. 

Let $\mathcal{L}' := \{ \pi_{\cv}(\xv) \colon \xv \in  \mathcal{B}_{\alpha}(\cv)\}$, where we view each $\pi_{\cv}(\xv) := \xv \land \cv$ as a vector in $S_\alpha^v$ (i.e., we ignore all coefficients $(\xv \land \cv)_i$ where $\cv_i = 0$). 
Note that the vectors in $\mathcal{L}'$ are independently and uniformly distributed over $S_\alpha^v$.  
That is, the problem of finding all $\xv', \yv' \in \mathcal{L}'$ satisfying $|\xv' \land \yv'| = e^*$ is a variant of an NNS problem (where the target weight $e^*$ does not necessarily equal the weight $\alpha$ of the vectors in the input list $\mathcal{L}'$), so we can apply the LSF techniques from \cite{DEEK23}. Furthermore, the set of those pairs $(\xv',\yv')$ (or actually the corresponding $(\xv, \yv)$) forms a superset of the set $\{(\xv, \yv) \in \mathcal{B}_\alpha(\cv)^2 \colon |\xv + \yv| = w\}$. 
Therefore, to asymptotically solve Problem~\ref{prob: bucket search}, it suffices to find all $\xv', \yv' \in \mathcal{L}'$ satisfying $|\xv' \land \yv'| = e^*$, and keep those for which it also holds that $|\xv + \yv| = w$. 

In order to asymptotically find all solution pairs in the bucket $\mathcal{B}_\alpha(\cv)$, we will thus focus on finding all pairs   $(\xv', \yv') \in \mathcal{L}' \times \mathcal{L}'$ satisfying  $|\xv' \land \yv'| = e^*$ and $|\xv + \yv| = w$. We will call such a pair  $(\xv', \yv')$ an \textit{$\mathcal{L}'$-solution}.

\begin{remark}[`Residual vectors for codes']
    For the reader familiar with the techniques from \cite{AC:ChaLoy21}: we consider the vectors in $\mathcal{L}'$ as a code analog of the `residual vectors' in \cite{AC:ChaLoy21}. 
    The lattice equivalent of $\subr{FindSolutions}$ aims to find all pairs $(\vv, \wv)$ in a given bucket that satisfy $\langle \vv, \wv \rangle = \theta$ for some given $\theta$, where the bucket is defined for some center $\cv\in\R^n$ (of unit norm) and bucketing parameter $\alpha \in (0,1)$, and consists of unit (lattice) vectors $\vv \in \R^n$ satisfying $\langle \vv, \cv \rangle \geq \alpha$. (Here, $\langle \cdot, \cdot \rangle$ denotes the Euclidean inner product.) Note that if $\vv$ is in the bucket of $\cv$, then it can be written as $\vv = \alpha \cv + \sqrt{1 - \alpha^2} \vv'$ for some unit vector $\vv'$ that is orthogonal to $\cv$. In particular, the problem of finding all $(\vv,\wv)$ satisfying $\langle \vv, \wv \rangle = \theta$ is equivalent to the problem of finding all  $(\vv',\wv')$ satisfying $\langle \vv', \wv' \rangle = \theta'$ for some $\theta'$ depending on $\theta$. These (left-over) vectors $\vv'$ are called \textit{residual vectors} in \cite{AC:ChaLoy21}, and are the analogs of the vectors $\xv' = \pi_{\cv}(\xv)$ that we defined for codes.  However, note that in the setting of codes, we don't have an exact equivalence between the condition $|\xv + \yv| = w$ and a condition $|\xv' + \yv'| = w'$ for some $w'$ depending on $w$ (unless $v = \alpha$). Nevertheless, Theorem~\ref{thm: existence residuals for codes} shows that most pairs $(\xv, \yv)$ satisfying $|\xv + \yv| = w$ will satisfy $|\xv' + \yv'| = w'$ for $w' := 2\alpha - 2e^*$. 
\end{remark}

It remains to prove  \Cref{thm: existence residuals for codes}. 
\begin{proof}[Proof of \Cref{thm: existence residuals for codes}]
Recall that $p$ denotes the probability that two independent and uniformly random $\xv,\yv \in \mathcal{B}_{\alpha}(\cv)$ satisfy $|\xv + \yv| = w$. From the definition of $p$, we observe that \begin{align*}
        p = \Pr_{\xv,\yv\in_R \mathcal{B}_\alpha(\cv)}[|\xv + \yv| = w] = \sum_{e = \max(0, 2\alpha - v)}^{\min(\alpha,w/2)} p(e)
    \end{align*}
where $p(e) := \Pr_{\xv,\yv\in_R \mathcal{B}_\alpha(\cv)}[|\xv + \yv| = w \text{ and } |\xv \land \yv \land \cv| = e]$. 
In particular, it holds that $p(e^*) \leq p \leq p(e^*) \cdot n/2$, where $e^* := \mathrm{argmax}_e p(e)$.\footnote{The value of $e^*$ can be computed numerically. For instance, see \cite{Car20, DEEK23}.} 
(Note that an explicit formula for $p(e)$ is given in \Cref{lem: probability p (wedge version)}.)
This implies that $\mathbb{E}[| \{ (\xv,\yv) \in \mathcal{B}_\alpha(\cv)^2 \colon |\xv + \yv| = w \} |]$ and $\mathbb{E}[| \{ (\xv,\yv) \in \mathcal{B}_\alpha(\cv)^2 \colon |\xv + \yv| = w \text{ and } |\xv \land \yv \land \cv| = e^* \} |]$ are equal up to a multiplicative factor that is linear in $n$, as we wanted to show.
\qed
\end{proof}

\subsection{On the Choice of $|\mathcal{C}_f'|$} \label{sec: sparsification} 

We will now explain how the choice of the size of the RPC $\mathcal{C}_f' \subseteq S_{v'}^v$ (together with the choice of $\beta$) affects the time complexity of our quantum-walk algorithm (\Cref{thm: bucket search using QW and RPC}) and present two choices of $|\mathcal{C}_f'|$ that we focus on. 
We start by showing how the two components $\mathbb{E}[|\mathcal{VC}_\beta(\xv)|]$ and $q$ in the time complexity depend on $|\Cs_f'|$. Since we identify the input list $\mathcal{B}_{\alpha}(\cv) \subseteq S_w^n$ with the list $\mathcal{L}' \subseteq S_\alpha^v$, we will from now on write $\mathcal{VC}_\beta(\xv')$ for $\xv' \in \mathcal{L}'$, instead of $\mathcal{VC}_\beta(\xv)$ for $\xv \in \mathcal{B}_{\alpha}(\cv)$.

For all $\xv' \in S_{\alpha}^{v}$, the probability that a uniformly random $\cv' \in S_{v'}^{v}$ satisfies $|\xv' \land \cv'| = \beta$ is \begin{align*}
        p_\beta := \Pr_{\cv' \in S_{v'}^{v}}[|\xv' \land \cv'| = \beta] = \frac{|\mathcal{C}_{v',\beta}^{v} (\xv')|}{|S_{v'}^{v}|}. 
\end{align*} 

Note that this also equals the probability that (for fixed $\cv' \in S_{v'}^{v}$) a uniformly random $\xv' \in S_{\alpha}^{v}$ satisfies $|\xv' \land \cv'| = \beta$. Therefore, $\mathbb{E}[|\mathcal{VC}_\beta(\xv')|] = |\mathcal{C}_f'| \cdot p_\beta $ (up to factors subexponential in $n$, see \cite[Lemma~4.6]{DEEK23}) and $\mathbb{E}[|\mathcal{B}_{\beta}(\cv')|] = s \cdot p_\beta$, where we recall that $s$ denotes the vertex size of the Johnson graph.\footnote{In the remainder of Section~\ref{sec: quantum speed-ups}, we often omit writing $\Tilde{O}(\cdot)$ or factors of the form $2^{o(n)}$ for ease of reading.}

Next, we will show how the probability $q$ depends on the size of $\Cs_f'$, where we recall that $q$ denotes the probability that a given solution pair can be found in a same $\beta$-bucket. 
First, note that we can express $q$ as follows.  
\begin{lemma}\label{lem: probability q}
    We have $q = \Pr[\exists \cv' \in \mathcal{C}_f' \text{ s.t. } \xv',\yv' \in \mathcal{B}_\beta(\cv')]$ 
    for all $\xv',\yv' \in \mathcal{L}'$ satisfying $|\xv' \land \yv'| = e^*$.
\end{lemma}
\begin{proof}
Notice that the following two probabilities are the same: 
\begin{itemize}
    \item For fixed $\xv',\yv' \in S$ satisfying $|\xv' \land \yv'| = e^*$, the probability that a uniformly random $\cv' \in S_{v'}^v$ satisfies $|\xv' \land \cv'| = \beta$ and $|\yv' \land \cv'| = \beta$. 
    \item For fixed $\xv',\yv' \in S$ satisfying $|\xv' \land \yv'| = e^*$ and $|\xv + \yv| = w$, the probability that a uniformly random $\cv' \in S_{v'}^v$ satisfies $|\xv' \land \cv'| = \beta$ and $|\yv' \land \cv'| = \beta$. 
\end{itemize}
Indeed, by a counting argument, one can show that both probabilities are equal to $\frac{|\mathcal{W}_{v', \beta}^{v}(\xv',\yv')|}{|S_{v'}^{v}|}$.
\qed
\end{proof} 

Fix an arbitrary pair $(\xv',\yv') \in \mathcal{L}'^2$ satisfying $|\xv' \land \yv'| = e^*$, and define \begin{align*}
            W := \Pr_{\cv' \in_R \mathcal{C}_f'}[|\xv' \land \cv'| = \beta \text{ and } |\yv' \land \cv'| = \beta] = \frac{|\mathcal{W}_{v', \beta}^{v}(\xv',\yv')|}{|S_{v'}^{v}|}.
        \end{align*}
Then we can write $q$ as  $q = 1 - (1 - W)^{|\mathcal{C}_f'|}$, so $q = \Theta(|\mathcal{C}_f'| W)$ if $|\mathcal{C}_f'| W \leq 1$, and $q = \Theta(1)$ otherwise. (Here, we are implicitly viewing $\mathcal{C}_f'$ as a random code; a formal justification for RPCs follows from \cite[Lemma~4.8]{DEEK23}.)

We will now consider two variants of the previous quantum-walk algorithm, where we vary the choice of $|\mathcal{C}_f'|$, i.e., the number of $\beta$-buckets, and see how that affects the two factors of the time complexity of our quantum-walk algorithm:
    $$\mathsf{S} = s \cdot \mathbb{E}[|\mathcal{VC}_\beta(\xv')|] \quad \text{ and } \quad \frac{1}{\sqrt{\epsilon}}\left(\frac{1}{\sqrt{\delta}}U+C\right) = \frac{1}{\sqrt{spq}} \left(\mathbb{E}[|\mathcal{VC}_\beta(\xv')|] + \sqrt{\mathbb{E}[|\mathcal{VC}_\beta(\xv')|] \cdot \mathbb{E}[|\mathcal{B}_\beta(\cv')|]}\right).$$ 
(Recall that the expected values should be replaced by 1 if they are smaller than 1, see Remark~\ref{rem: writing max(1, exp size) as exp size}.)

\subsubsection{Variant 1: Choosing $|\mathcal{C}_f'|$ such that $q$ is maximized.}
A first approach is to choose $|\mathcal{C}_f'|$ such that each $\mathcal{L}'$-solution $\mathbf{x}', \mathbf{y}'$ have at least one bucket in common, i.e., there are no false-negatives. In particular, for similar reasons as in the proof of \cite[Theorem~4.4]{DEEK23} (but then considered for our \textit{second} layer of LSF) taking $|\mathcal{C}_f'| = \Omega(\frac{1}{W})$ ensures $q = \Theta(1)$ and $\mathbb{E}[|\mathcal{VC}_\beta(\xv')|] = \frac{|\mathcal{C}_{v',\beta}^{v} (\xv')|}{|\mathcal{W}_{v', \beta}^{v}(\xv',\yv')|} = \frac{p_\beta}{W}$ (which is $\geq 1$).

Then the two contributors to the quantum-walk cost are as follows: 
\begin{align}\label{eq: QW params without sparsification}
    \mathsf{S} = s \cdot \frac{p_\beta}{W} \quad \text{ and } \quad
    \frac{1}{\sqrt{\epsilon}}\left(\frac{1}{\sqrt{\delta}} \mathsf{U} + \mathsf{C} \right) = \sqrt{\frac{p_\beta}{spW}} \cdot \left( \sqrt{\frac{p_\beta}{W}} + \max(1, \sqrt{s p_\beta}) \right). 
\end{align}

\subsubsection{Variant 2: Sparsification.} 
A way to obtain a better time complexity is to apply a technique that we refer to as \textit{sparsification}, which was already used in \cite{Laa16, AC:ChaLoy21}. 
Reducing the size $|\mathcal{C}_f'|$ (taking a \textit{sparser} code $\mathcal{C}_f'$), reduces both $q$ and $\mathbb{E}[|\mathcal{VC}_\beta(\mathbf{x}')|]$, resulting in an increase in the cost $1/\sqrt{\epsilon}$, but a decrease in both the set-up $\mathsf{S}$ and update $\mathsf{U}$ costs. Altogether, this results in a different balance between the terms $\mathsf{S}$ and $\frac{1}{\sqrt{\epsilon}}(\frac{1}{\sqrt{\delta}} \mathsf{U} + \mathsf{C})$  in the runtime of the quantum walk subroutine.

Following the same reasoning as in \cite{AC:ChaLoy21}, we take $|\mathcal{C}_f'| = \frac{|S_{v'}^{v}|}{|\mathcal{C}_{v',\beta}^{v} (\xv')|} = \frac{1}{p_\beta}$. 
Note that now $q = \Omega(\frac{W}{p_\beta})$ and 
$\mathbb{E}[|\mathcal{VC}_\beta(\xv')|] = 1$ (recall that we omit writing $2^{o(n)}$), indeed resulting in a different balance between the contributors of the quantum-walk cost: 
\begin{align}\label{eq: QW params with sparsification}
    \mathsf{S} = s \quad \text{ and } \quad \frac{1}{\sqrt{\epsilon}}\left(\frac{1}{\sqrt{\delta}}\mathsf{U} + \mathsf{C}\right) = \sqrt{\frac{p_\beta}{spW}}(1 + \max(1, \sqrt{sp_\beta})).
\end{align}
More precisely, comparing Equation~\eqref{eq: QW params with sparsification} to Equation~\eqref{eq: QW params without sparsification}, shows that sparsification has reduced the set-up cost $\mathsf{S}$ by a factor of $\frac{p_\beta}{W}$ and part of the cost $\frac{1}{\sqrt{\epsilon}}\left(\frac{1}{\sqrt{\delta}}\mathsf{U} + \mathsf{C}\right)$ by a factor of $\sqrt{\frac{p_\beta}{W}}$.

\

We will present further in Section~\ref{sec: numerical results} our numerical results of these formulas (optimized to minimize the time complexity), which illustrate that the second choice -- i.e., sparsification -- indeed provides a speed-up over the first.

\subsection{Proof of the Complexity of \subr{FindSolutions} Using Quantum Walks}\label{sec: proof of bucket search using QW}

We will now prove \Cref{thm: bucket search using QW and RPC}. Recall  (see Remark~\ref{rem: writing max(1, exp size) as exp size}) that we write $\mathbb{E}[|\mathcal{VC}_{\beta}(\mathbf{x})|])$ and $\mathbb{E}[|\mathcal{B}_{\beta}(\mathbf{c}')|]$ as shorthand for $\max(1, \mathbb{E}[|\mathcal{VC}_{\beta}(\mathbf{x})|])$ and $\max(1, \mathbb{E}[|\mathcal{B}_{\beta}(\mathbf{c}')|])$, respectively.

\subsubsection{Time complexity analysis}

\begin{proof}[Proof of \Cref{thm: bucket search using QW and RPC}: Part 1/2, time complexity]
Let $\mathcal{B}_\alpha(\cv)$ be an instance of Problem~\ref{prob: bucket search} with parameters $(B, n, w, v, \alpha)$, and let $p = p(n,w,v,\alpha)$ be according to \Cref{lem: probability p (wedge version)}. 

We start by describing the details of the quantum walk according to the MNRS framework \cite{MNRS11}. 
We identify $\mathcal{B}_\alpha(\cv)$ with the list $\mathcal{L}' \subseteq S_\alpha^v$ defined in \Cref{sec: reduction to a variant of NNS}, and recall that each vector $\xv \in \mathcal{B}_\alpha(\cv)$ has an associated  $\xv' \in \mathcal{L}'$ (more precisely, $\mathcal{L}' = \{ \xv \land \cv \colon \xv \in  \mathcal{B}_{\alpha}(\cv)\}$ where $\xv \land \cv$ is viewed as a vector in $S_\alpha^v$).  
By the arguments from \Cref{sec: reduction to a variant of NNS} (and for the $e^*$ defined there), it suffices to instead search for  $\mathcal{L}'$-solutions, i.e., pairs $(\xv', \yv') \in \mathcal{L}' \times \mathcal{L}'$ satisfying $|\xv' \land \yv'| = e^*$ and $|\xv + \yv| = w$. 
We will find those $\mathcal{L}'$-solutions by applying a quantum walk using an extra layer of LSF. 
More precisely, our walk will be over subsets $S$ of $\mathcal{L}'$, and we aim to reach a vertex $S$ that contains an $\mathcal{L}'$-solution. This search for a vertex containing an $\mathcal{L}'$-solution will be facilitated by the use of bucketing methods; however, this time we do both the bucketing and checking phase within the quantum walk. 

\paragraph{Graph.} Consider a quantum walk on the Johnson graph $J(|\mathcal{L}'|, s) = (V,E)$ for some suitable parameter $s < |\mathcal{L}'| = B$ satisfying $s = \Tilde{O}(\frac{1}{\sqrt{p}})$.\footnote{This constraint on $s$ ensures that the number of $\mathcal{L}'$-solutions per vertex is $\Tilde{O}(1)$ on average. It is used in our analysis of the set-up and update costs, as well as in our derivation of $\epsilon$.} 
In other words, the set of vertices is $V = \{ S \subseteq \mathcal{L}' \colon |S| = s \}$ and two vertices $S_1, S_2$ are adjacent (i.e., $(S_1, S_2) \in E$) if and only if $|S_1 \cap S_2| = s - 1$. 

\paragraph{Data.} To each vertex $S \in V$ we add some additional data structure $\textsc{data}(S)$ that allows us to efficiently check in the update step whether a newly added element forms an $\mathcal{L}'$-solution with one of the existing elements in a vertex (instead of having to search for a `colliding element' among all $s-1$ other vertex elements). 
More precisely, the idea (originating from \cite{AC:ChaLoy21}) is as follows. We will invoke the RPC framework from \cite{DEEK23} (\Cref{sec: RPC}) again and sample a random product code $\mathcal{C}_f' \sim R_{v, v', t', C}$ (i.e., $\mathcal{C}_f' \subseteq S_{v'}^v$) for $t' = \Theta(\sqrt{v})$ with $v'$ to be determined. 
Let $\beta$ be another parameter to be determined later. Instead of searching for any $\mathcal{L}'$-solution in $S$, the idea is to only search for solutions that both have overlap $\beta$ with a bucket center from $\mathcal{C}_f'$.\footnote{This could mean that there are some false negatives: it might happen that $S$ contains an $\mathcal{L}'$-solution, but that it isn't captured in one of the buckets. However, by choosing the parameters $|\mathcal{C}_f'|$ and $\beta$ carefully we can control how likely this is to happen. (For more details, see the explanation in our derivation of $\epsilon$.)}
Therefore, we will add the following data $\textsc{data}(S)$ to each vertex $S \in V$: 
\begin{itemize}
    \item The buckets $\mathcal{B}_{\beta}(\cv') := \{ \xv' \in S \colon |\xv' \land \cv'| = \beta\}$ for each $\cv' \in \mathcal{C}_f'$. These will be `stored' using a data structure that allows for efficient inserting and removing of elements. (See \cite{Amb07, BJLM13} for details.) 
    \item To each bucket $\mathcal{B}_{\beta}(\cv')$, we will add a list containing the $\mathcal{L}'$-solutions in the bucket, if there are any. 
    \item We keep track of a bit indicating whether $S$ is marked according to the following definition. 
\end{itemize} 

\paragraph{Marked vertices.} Define the set $M \subseteq V$ of marked vertices as
\begin{align*}
    M = \{S \in V \colon \exists\, \cv' \in \mathcal{C}_f', \exists\, \xv', \yv' \in \mathcal{B}_{\beta}(\cv') \text{ s.t. } (\xv',\yv') \text{ is an $\mathcal{L}'$-solution} \}.
\end{align*} 
Note that this is a subset of the set of vertices containing an $\mathcal{L}'$-solution. 
In particular, an $\mathcal{L}'$-solution will only be detected if it (the pair) can be found in a $\mathcal{C}_f'$-bucket. 

We will now describe the parameters determining the expected time complexity of the quantum walk; the memory complexity will be discussed after. We will repeatedly make use of the fact that, for any $\xv' \in \mathcal{L}'$, the set $\mathcal{VC}_\beta(\xv') := \{\cv' \in \mathcal{C}_f' \colon |\xv' \land \cv'| = \beta\}$ of its valid buckets can be computed in time $|\mathcal{VC}_\beta(\xv')| + 2^{o(n)}$ by \Cref{thm: properties RPC}. 
Furthermore, for ease of presentation, we will often omit writing down $\Tilde{O}(\cdot)$ and factors of the form $2^{o(n)}$.  

\paragraph{$\mathsf{S}$, set-up cost.}  
Constructing a uniform superposition over all $S \in V$ costs $\Tilde{O}(s)$. In addition, we need to compute the data for each vertex $S$. We will construct the quantum data structure storing the buckets $\mathcal{B}_{\beta}(\cv')$ by computing, for each $\xv' \in S$, the set $\mathcal{VC}_\beta(\xv')$ of valid buckets, and then insert $\xv'$ to the right locations in the quantum data structure. Since insertion can be done efficiently, for instance using one of the quantum data structures in \cite{Amb07, BJLM13}, this takes  $\Tilde{O}(s \cdot \mathbb{E}[|\mathcal{VC}_\beta(\xv')|])$ time in expectation. 
It remains to construct the second and third part of the data. Note that the condition on $s$ implies that the expected number of $\mathcal{L}'$-solutions per vertex is ${s \choose 2} p = \Tilde{O}(1)$. Therefore, applying Grover over all pairs in $S$ allows for finding these (in expectation) ${s \choose 2} p = \Tilde{O}(1)$ pairs in time $\Tilde{O}(s)$. 
For any found pair $(\xv', \yv')$, it then suffices to go over the valid $\beta$-buckets of $\xv'$ to add $(\xv', \yv')$ to any bucket that also contains $\yv'$.\footnote{This step could possibly be sped up, but there is no reason to as it does not dominate the set-up cost.} 
If at least one pair is added, we add $1$ as third component of the data, and $0$ otherwise. Thus, altogether the construction of this second and third part of the data cost at most $\Tilde{O}(s + \mathbb{E}[|\mathcal{VC}_\beta(\xv')|])$ time. 
We conclude that $\mathsf{S} = \max(s, s \cdot \mathbb{E}[|\mathcal{VC}_\beta(\xv')|], s + \mathbb{E}[|\mathcal{VC}_\beta(\xv')|]) = s \cdot \mathbb{E}[|\mathcal{VC}_\beta(\xv')|]$.

\paragraph{$\mathsf{C}$, checking cost.}
Since we only need to check the last component of the data, we have $\mathsf{C} = 1$.

\paragraph{$\mathsf{U}$, update cost.} 
The dominating cost will be that of updating the data when going from a vertex $S$ to a neighbor $S'$. Updating the first part of the data, i.e., the buckets, can be done in time $\mathbb{E}[|\mathcal{VC}_\beta(\xv')|]$, since it suffices to compute the sets of valid buckets for the old and new element, to remove the old element from its valid buckets, and to add the new element to its valid buckets. Updating the second part of the data (and the third part, if needed) can be done in time $\sqrt{\mathbb{E}[|\mathcal{VC}_\beta(\xv')|] \cdot \mathbb{E}[|\mathcal{B}_{\beta}(\cv')|]}$: apply Grover to the union of all valid buckets of the old, respectively new, element to see if they are part of an $\mathcal{L}'$-solution; if so, remove, respectively add, this $\mathcal{L}'$-solution. (Here, we again make use of the fact that the condition on $s$ implies that the number of solutions per vertex is $s \cdot p \leq s \cdot \sqrt{p} = \Tilde{O}(1)$.)
It follows that $\mathsf{U} = \mathbb{E}[|\mathcal{VC}_\beta(\xv')|] +  \sqrt{\mathbb{E}[|\mathcal{VC}_\beta(\xv')|] \cdot \mathbb{E}[|\mathcal{B}_{\beta}(\cv')|]}$.

\paragraph{$\delta$, spectral gap.}
As mentioned in \Cref{sec: preliminaries}, the spectral gap of the Johnson graph satisfies $\delta = \Omega(\frac{1}{s})$.

\

For the derivation of $\epsilon$, we use the following lemma. 

\begin{lemma}\label{lem: derivation of epsilon for QW with RPC}
Let $\mathcal{C}_f' \sim R_{v,v',t', C}$ and let $S$ be subset of $\mathcal{L}'$ of size $s$ sampled independently and uniformly at random. 
Write $E_1$ for the event that there exists $\cv' \in \mathcal{C}_f'$ with $\xv', \yv' \in \mathcal{B}_{\beta}(\cv') := \{\xv' \in S \mid |\xv' \land \cv'| = \beta\}$ satisfying $|\xv' \land \yv'| = e^*$ and $|\xv + \yv| = w$ (i.e., the $\beta$-bucket of $\cv'$ contains an $\mathcal{L}'$-solution). Then \begin{align*}
    \Pr[E_1] \geq \min(q, \Omega(q s^2 p))
\end{align*} 
where $q := \Pr[\exists \cv' \in \mathcal{C}_f' \text{ s.t. } \xv',\yv' \in \mathcal{B}_\beta(\cv')]$ for an arbitrary $\xv',\yv' \in \mathcal{L}'$ satisfying $|\xv' \land \yv'| = e^*$. 
\end{lemma} 

\begin{proof}[Proof of \Cref{lem: derivation of epsilon for QW with RPC}]
Define $E_0$ to be the event that there exist $\xv', \yv' \in S$ satisfying $|\xv' \land \yv'| = e^*$ and $|\xv + \yv| = w$. Note that $\Pr[E_0] = \min(1, \Theta(s^2 p(e^*))) = \min(1, \Theta(s^2 p))$. Therefore, we have that 
\begin{align*}
        \Pr[E_1] = \frac{\Pr[E_1 \mid E_0] \Pr[E_0]}{\Pr[E_0 \mid E_1]} \geq \Pr[E_1 \mid E_0] \Pr[E_0] = \Pr[E_1 \mid E_0] \min(1, \Theta(s^2 p)).
\end{align*}
So it remains to show that $\Pr[E_1 \mid E_0] \geq \Omega(q)$.  
So suppose that $E_0$ holds. Then there exists a pair $\xv',\yv' \in S$ satisfying $|\xv' \land \yv'| = e^*$ and $|\xv + \yv| = w$. It then follows that the probability that $E_1$ holds (conditional on $E_0$) is at least the probability that there exists $\cv' \in \mathcal{C}_f'$ such that this particular pair $\xv',\yv'$ is in the bucket of $\cv'$. 
In other words, 
\begin{align*}
    \Pr[E_1 \mid E_0] \geq \Pr[\exists \cv' \in \mathcal{C}_f' \text{ s.t. } \xv',\yv' \in \mathcal{B}_\beta(\cv')]
\end{align*}
for any $\xv',\yv' \in S$ satisfying $|\xv' \land \yv'| = e^*$ and $|\xv + \yv| = w$. 
By \Cref{lem: probability q}, we have that $\Pr[\exists \cv' \in \mathcal{C}_f' \text{ s.t. } \xv',\yv' \in \mathcal{B}_\beta(\cv')] = q$ from which it follows that $\Pr[E_1 \mid E_0] = \Omega(q)$, finishing the proof. 
\qed
\end{proof}

\paragraph{$\epsilon$, probability of a vertex being marked.} 
Note that $\epsilon = \Pr[E_1]$ for $E_1$ as defined in the statement of \Cref{lem: derivation of epsilon for QW with RPC}. It follows from the lemma that $\epsilon = \min(q, \Omega(q {s \choose 2} p))$, where $q := \Pr[\exists \cv' \in \mathcal{C}_f' \text{ s.t. } \xv',\yv' \in \mathcal{B}_\beta(\cv')]$ for any $\xv',\yv'$ satisfying $|\xv' \land \yv'| = e^*$.
Our condition on $s$ implies that ${s \choose 2} p = \Tilde{O}(1)$, and thus $\epsilon = \Tilde{\Omega}(q s^2 p)$.

\paragraph{Conclusion on expected runtime.}
Since the number of solution pairs in a bucket $\mathcal{B}_{\alpha}(\mathbf{c})$ is in expectation $t = O(|\mathcal{B}_{\alpha}(\mathbf{c})|^2 p)$, the expected runtime of the quantum walk, repeated  $t$ times, is as given in the theorem statement. 
\qed 
\end{proof}

\begin{remark}
    The proof can be adapted to relax the condition $s = \Tilde{O}(1/\sqrt{p})$ to $s = \Tilde{O}(1/\sqrt{pq})$ if we replace $\epsilon = s^2 pq$ by $\epsilon = \min(q, s^2 pq)$. However, when we make these changes in our implementation of the algorithm, a numerical optimization of the parameters does not result in better complexities. Therefore, we only present the proof for  $s = \Tilde{O}(1/\sqrt{p})$ here. 
\end{remark}

It remains to prove the claims on the memory complexity. 

\subsubsection{Memory complexity analysis}

\begin{proof}[Proof of \Cref{thm: bucket search using QW and RPC}: Part 2/2, memory complexity]

\paragraph{Classical memory.}
The algorithm stores the bucket $\mathcal{B}_\alpha(\cv)$ (or, equivalently, the list $\mathcal{L}'$), so the algorithm uses classical space $\Tilde{O}(|\mathcal{B}_\alpha(\cv)|)$.

\paragraph{QRACM.}
The quantum walk uses quantum access to the bucket $\mathcal{B}_\alpha(\cv)$, stored classically and of size $M_{QRACM} = \Tilde{O}(|\mathcal{B}_\alpha(\cv)|)$. 

\paragraph{Quantum memory and QRAQM.}
The quantum walk stores a superposition of vertices. Each vertex contains $s$ elements and some additional data, including the buckets of the second layer. (The other components of the data will not dominate the QRAQM cost, so we can safely ignore those.) Within each vertex, a vector $\xv'$ is inserted in $|\mathcal{VC}_\beta(\mathbf{x}')|$ buckets. Altogether, this is stored in a quantum register of size $M_Q = s + s \cdot \mathbb{E}[|\mathcal{VC}_\beta(\mathbf{x}')|]$ (which is the number of qubits used). Since we use the data structures from \cite{Amb07} (or \cite{BJLM13}) to efficiently delete and insert vectors in superposition, we use QRAQM of size $M_{QRAQM} = M_Q$.  
\qed
\end{proof}

\section{Numerical Results} 
\label{sec: numerical results}
\label{sec:numerics}
In this section, we summarize our numerical results on the asymptotic runtime and memory of four different algorithms for the near-neighbor search (NNS) problem defined in Section~\ref{sec: sieving for codes} (Problem~\ref{prob: NNS}), and consequently for code sieving. As our work is inspired by its lattice-based analog, we complete the analysis by also including the complexities of the lattice-based equivalents of the four algorithms. The algorithms we compare are the following:
\begin{itemize}
    \item \textsc{Classical}: The sieving algorithm originally introduced in \cite{DEEK23}. It is based on the RPC approach described in Section~\ref{sec: RPC}, which is the best known approach in the classical case. Its lattice-based analog is presented in \cite{BDGL16}.   
    \item  \textsc{Grover}: The first quantum algorithm introduced in our paper (\cref{sec:grover}). It is a natural modification from the classical RPC algorithm to the quantum model, obtained using Grover's algorithm. A lattice-based analog was introduced in \cite{Laa16}.
    \item \textsc{QW + LSF}: A more general quantum approach based on quantum walks in combination with the additional layer of locality-sensitive filtering, as explained in \cref{sec:quwalk}. Its lattice-based analog was introduced in \cite[Sec. 4]{AC:ChaLoy21}.
    \item \textsc{QW + LSF + Sparsification}: A variant of our quantum-walk algorithm obtained using sparsification, as explained in \cref{sec: sparsification}. The lattice-based analog is given in \cite[Sec. 5]{AC:ChaLoy21}.
\end{itemize}

\Cref{fig: Runtime comparison} shows the asymptotic runtime of the four algorithms in the code-based setting. The figures are obtained by calculating the asymptotic runtime in $100$ equidistant values of $\omega := w/n$ ranging in $[0, 0.5)$, where $n,w$ are parameters of Problem~\ref{prob: NNS} and $w$ is chosen such that there is a unique solution to the problem on average (i.e., we analyze the problem in the unique-decoding regime). 

\begin{figure}
    \centering 
    \includegraphics[width=0.8\textwidth]{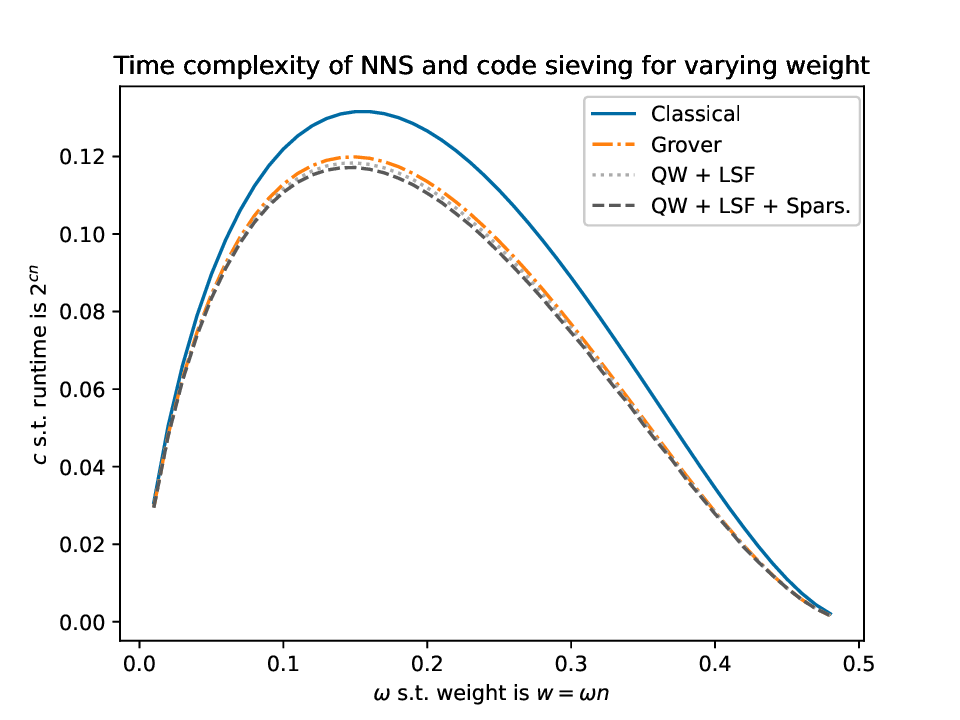}
    \caption{Comparison of the asymptotic runtime of the four algorithms solving $\NNS(n,w,N)$ for $w = \omega \cdot n$, where $\omega \in [0,0.5)$.}
    \label{fig: Runtime comparison}
\end{figure}

\newpage
\Cref{tab:HardestInstances} shows the calculated asymptotic runtime and memory for the hardest instances of NNS in the code-based and lattice-based settings, where by the \textit{hardest instances} we mean those values of $\omega$ for which the runtime curve in \Cref{fig: Runtime comparison} reaches its peak. The asymptotic memory is then calculated for the same value of $\omega$. As we distinguish four different types of memory, we denote them by $M_C, M_Q, M_{QRACM}, M_{QRAQM}$ and refer to them as classical memory, quantum memory, QRACM, and QRAQM, respectively.


\begin{table}[h!]
\begin{center}
	\begin{tabular}{| m{1.5cm}  m{5cm} | m{1.5cm} m{1.5cm} m{1.5cm} m{1.5cm} m{1.5cm} |}
        \hline
		& Algorithms & $t$ & $M_C$ & $M_Q$ & $M_{QRACM}$ & $M_{QRAQM}$ \\
		\hline
            \hline
            \multirow{4}{*}{Codes} & \textsc{Classical} \cite{DEEK23} & 0.132 & 0.093 & - & - & - \\
			& \textsc{Grover} & 0.120 & 0.094 & 0 & 0.026 & - \\
			& \textsc{QW + LSF} & 0.118 & 0.094 & 0.024 & 0.031 & 0.024 \\
			& \textsc{QW + LSF + Spars.} & 0.117 & 0.094 & 0.023 & 0.036 & 0.023 \\
		\hline
            \hline
            \multirow{4}{*}{Lattices} & \textsc{Classical} \cite{BDGL16} & $0.292$ & $0.208$ & - & - & - \\
			& \textsc{Grover} \cite{Laa16} & $0.265$ & $0.208$ & $0$ & $0.058$ & - \\
			& \textsc{QW + LSF} \cite{AC:ChaLoy21} & $0.261$ & $0.208$ & $0.053$ & $0.069$ & $0.053$ \\
			& \textsc{QW + LSF + Spars.} \cite{AC:ChaLoy21} & $0.257$ & $0.208$ & $0.050$ & $0.077$ & $0.050$ \\ 
		\hline
		\end{tabular}
\end{center}
\caption{The exponents of the asymptotic runtime (i.e., $t$ s.t.\ runtime is $2^{t n + o(n)}$) and corresponding memory exponents for the hardest instances of NNS for codes (i.e., for those $\omega$ that yield the highest runtime for each considered algorithm), and the asymptotic runtime and memory of their lattice analogs. A dash `-' means that the memory type is not used.}
\label{tab:HardestInstances}
\end{table}

\begin{remark}
    We emphasize that \Cref{tab:HardestInstances} does not provide the asymptotic complexity of information-set decoding algorithms, in contrast to what was presented in \cite{DEEK23}. In particular, the number 0.132 in the table refers to the runtime of the classical \textit{sieving} algorithm from \cite{DEEK23} (here referred to as \textsc{Classical}), and is not explicitly stated in their paper.    
\end{remark}

All numerical results were obtained using Python code, available in the aforementioned repository.

\begin{remark}\label{rem: reusable walk}
    Our repository also includes code to estimate the complexity of a variant of \textsc{QW+LSF+} \textsc{Sparsification} where we incorporate the reusable quantum-walk techniques from \cite{BCSS23} (in a similar way as they apply their techniques to \cite{AC:ChaLoy21}). Similar to their result on lattice sieving, we only obtain a minor improvement, namely an exponent of $t = 0.1169$ instead of $0.1171$. (In lattice sieving, \cite{BCSS23} obtain an  exponent of $t = 0.2563$ instead of $0.2570$.)
    We leave the elaboration of the analytical details of this approach as potential future work. 
\end{remark}

\subsubsection{Some observations.}

Our numerical results show that our quantum algorithms provide a speed-up in comparison with the classical runtime \cite{DEEK23}. While this is non-surprising, it is not necessarily guaranteed. Despite the structural differences between codes and lattices (particularly, their metric), we observe that similar techniques yield comparable improvements in the asymptotic runtime for sieving in both contexts. Grover's algorithm applied to the classical version delivers the most substantial quantum speed-up. Nevertheless, our quantum-walk algorithms outperform the version using Grover, with a slight improvement obtained using the sparsification technique. 
It appears that fundamentally different quantum techniques are needed to obtain more significant improvements. 

While our work only focused on improving the asymptotic time complexities of sieving algorithms, it is also important to quantify the memory used by these algorithms, and minimize it. 
In particular, note that the improved runtime obtained when moving from classical to Grover, and subsequently to the quantum-walk algorithms, comes at the cost of an increased memory complexity and extra assumption about the memory model. Specifically, the Grover version uses quantum memory and QRACM access; quantum walks additionally use QRAQM, which is an even stronger assumption about the memory model. 
Note that this trade-off between improved runtime and assumption regarding the memory model is consistent with the observations from other quantum algorithms in code-based (and lattice-based) cryptanalysis. 

\section{On the Application to Quantum Information-Set Decoding}
\label{sec: sievingISD}
Information-set decoding (ISD) algorithms are known as the most efficient generic attacks against the decoding problem for a wide range of parameters.\footnote{For a more specific range of parameters, there exist more efficient attacks such as statistical decoding \cite{AC:CDMT22}.}
Each of these can be seen as improvements of the algorithm originally introduced by Prange\cite{Pra62}. We view ISD algorithms as a framework, as formulated in \cite{AC:FinSen09}. Specifically, following \cite{EPRINT:GuoJohNgu23} and \cite{DEEK23} we see the sieving algorithm as a subroutine of an ISD algorithm within the framework. We refer to the resulting algorithm as \textit{SievingISD}.  
In this section, we analyze a quantum analog of SievingISD.

\begin{remark}
    In this section, the parameters $n,w$ have a different meaning than in previous sections.
\end{remark}

\subsection{ISD and SievingISD}

The central computational problem underlying the security of code-based primitives is the aforementioned decoding problem (Problem~\ref{prob: decoding problem, N sol}) with $N = 1$. For completeness, we restate the $N = 1$ variant here. We are primarily interested in the complexity of the cryptographically relevant instances of this problem. Namely, we analyze the problem for $\Cs$ being a random $[n,k]$ code, and $w$ chosen such that there is only one solution to the problem on average. We refer to these instances as a \textit{unique decoding instance} with parameters $(n,k,w)$.

\begin{problem}[Decoding problem, $\DP(n,k,w)$]\label{prob: decoding problem, 1 sol}
Given an $[n,k]$ binary linear code $\Cs$ and an integer value $w$, find a codeword $\xv_{\Cs} \in \Cs$ of weight $|\xv_{\Cs}| := w$. 
\end{problem}

Algorithm~\ref{alg: ISD} presents the ISD framework, using the same formulation as in \cite{DEEK23}. As part of the input, the algorithm is given an oracle $\oracleA$ that, given a $[n', k]$ binary linear code $\Cs'$ and an integer $w'$, returns $N$ independent and uniformly random weight-$w'$ codewords in $\Cs'$. We use the notation $\pi_{\xv}(\cdot)$ for code puncturing, which is defined as follows.

\begin{definition}[Code puncturing, $\pi_{\xv}(\cdot)$]
For a linear $[n,k]$ code $\Cs$ and a binary vector $\xv \in \F_2^n$ with $|\xv|=n'$ we define by $\pi_{\xv}: \cv \mapsto \cv \wedge \xv$ the puncturing function relative to the support of $\xv$, and we define $\pi_{\xv}(\Cs)$ to be the corresponding punctured code, which is a $[n',k]$ binary linear code.
\end{definition}

\begin{algorithm}
\caption{Information-set decoding (ISD)}\label{alg: ISD}
\DontPrintSemicolon
\SetKwInOut{Input}{Input}\SetKwInOut{Output}{Output}

\Input {$[n,k]$ linear code $\mathcal{C}$, weight $w$, and an oracle $\oracleA$ as described above (for fixed $n' > k$, $w'$, $N$)}
\Output{$\mathbf{c} \in \mathcal{C}$ such that $|\mathbf{c}| = w$}
\vspace{0.2cm}
\While{True}{
Choose $\mathbf{x} \in \Ss_{n'}^n$ uniformly at random and repeat if $\dim(\pi_{\mathbf{x}}(\mathcal{C})) \neq k$\; 
$\mathcal{L} \leftarrow \mathcal{A}(\pi_{\mathbf{x}}(\mathcal{C}), w')$\;
\If{$\exists~ \mathbf{y} \in \mathcal{L} \colon |\pi_{\mathbf{x}}^{-1}(\mathbf{y})| = w$}{
\Return $\pi_{\mathbf{x}}^{-1}(\mathbf{y})$}
}
\end{algorithm}

We remark that for a uniformly random $\xv \in \F_2^n$ with $|\xv| = n'$, it happens with constant probability that $\dim(\pi_{\xv}(\mathcal{C})) = k$ \cite{Cooper2000}. 

We refer to \textit{SievingISD} as any ISD algorithm in the form of Algorithm~\ref{alg: ISD} that uses a sieving algorithm (Algorithm~\ref{alg: sieving}) as input oracle $\mathcal{A}$. Recall that by \Cref{lem: lower bound on N for NNS} and the subsequent discussion that this requires $N = \Omega\left(\frac{{n' \choose w'}}{{w' \choose w'/2} {n' - w' \choose w'/2}}\right)$.

\subsection{Quantum ISD and Quantum SievingISD} 

We will now discuss quantum analogs of the ISD framework and SievingISD. Quantum ISD algorithms (e.g., \cite{PQCRYPTO:KacTil17, PQCRYPTO:Kirshanova18}) are based on the idea that one can apply amplitude amplification (AA) \cite{BHMT02} to speed up the search for the solution of the decoding problem over multiple iterations of the ISD. Furthermore, the subroutine $\mathcal{A}$ is allowed to be quantum. 
It results in the following quantum analog of \cite[Theorem~3.1]{DEEK23}. Here, and in the remainder of Section~\ref{sec: sievingISD}, we only focus on the runtime of (quantum) ISD algorithms, not their space usage.

\begin{theorem}[Quantum ISD] \label{thm: complexity of quantum ISD assuming existence oracle}
    Let $\mathcal{C}$ be a unique decoding instance with parameters $(n,k,w)$. Let $n' > k$ and let $w' \leq n'$. 
    Suppose $\mathcal{A}$ is an algorithm that returns $N$ independent and uniformly random weight-$w'$ codewords in a given $[n', k]$ binary linear code, where $N \leq {n' \choose w'}/2^{n' - k}$. 
    \\
    If the expected runtime of $\mathcal{A}$ is $T_{\mathcal{A}}$, 
    then there is a quantum algorithm that returns a weight-$w$ codeword in $\mathcal{C}$, if one exists, in expected time \begin{align*}
        \Tilde{O}\left(\frac{T_{\mathcal{A}}}{\sqrt{p_1 p_2}}\right)  
    \end{align*} 
    where $p_1 := \frac{{n' \choose w'} {n - n' \choose w - w'}}{{n \choose w}}$ and $p_2 := \frac{N \cdot 2^{n' - k}}{{n' \choose w'}}$. 
\end{theorem}

\begin{proof} 
(This result is not new, but we sketch the proof for completeness.) 
The correctness follows from the proof of Theorem 3.1 in \cite{DEEK23}. Moreover, since the success probability of the ISD algorithm (Algorithm~\ref{alg: ISD}) is given by $p_1 p_2$, amplitude amplification (\Cref{thm: AA}) allows the algorithm to succeed after $\frac{1}{\sqrt{p_1 p_2}}$ iterations. One iteration is dominated by the time it takes for algorithm $\mathcal{A}$, which is~$T_{\mathcal{A}}$. 
\qed
\end{proof}

\begin{remark}\label{rem: upper bound on N}
    We recall from \Cref{sec: sieving for codes} that the upper bound on $N$ is to ensure that there exist $N$ codewords on average (for a random code $\mathcal{C}$) as output of Algorithm~\ref{alg: sieving}. Note that it ensures $p_2 \leq 1$.  
\end{remark}

If $N = \Omega\left({n' \choose w'}/\left({w' \choose w'/2} {n' - w' \choose w'/2}\right)\right)$, we can use a  (classical or quantum) sieving algorithm as the subroutine $\mathcal{A}$, resulting in a natural quantum analog of SievingISD. We refer to it as \textit{quantum SievingISD}. It turns out that quantum SievingISD does not allow for runtime close to the best known quantum ISD algorithms, as we will now show.

\subsection{Limitations of Quantum SievingISD}\label{sec: limitations of quantum SievingISD}

We numerically illustrate that, contrary to the classical setting, \textit{quantum} SievingISD cannot do better than what we will refer to as \textit{quantum Prange}, the classical Prange algorithm quantized with AA due to \cite{Bernstein10}. Quantum Prange was the first quantum ISD algorithm, and therefore a natural starting point for comparison.\footnote{We remark that the more recent quantum ISD algorithms in \cite{PQCRYPTO:KacTil17, PQCRYPTO:Kirshanova18} have even better time complexities than quantum Prange.}  
We recall its time complexity. 

\begin{lemma}[Quantum Prange, \cite{Bernstein10}]\label{lem: quantum Prange}
    Consider a unique decoding instance of DP with parameters $n, k, w$. Then there is a quantum algorithm that solves it in time \begin{align*}
    \Tilde{O}\left(\sqrt{\frac{{n \choose w}}{{n - k \choose w}}}\right). 
\end{align*}
\end{lemma}

Let us first formalize our claim. Recall that any (classical or quantum) SievingISD algorithm requires $N$ to satisfy \begin{align}\label{eq: bounds on N}
        N = \Omega\Bigg(\frac{{n' \choose w'}}{{w' \choose \frac{w'}{2}} {n'-w' \choose \frac{w'}{2}}}\Bigg) \quad \text{ and } \quad N = O\left(\frac{{n' \choose w'}}{2^{n' - k}}\right).
    \end{align}
While the ISD framework itself imposes the upper bound (see \Cref{rem: upper bound on N}), the lower bound is imposed due to using a \textit{sieving} subroutine (recall \Cref{lem: lower bound on N for NNS}). The lower bound turns out to be a bottleneck in being able to improve over quantum Prange, as the following claim (and its justification) illustrates. Since our justification of the claim is partly numerical, we do not refer to it as a `theorem'.

\begin{claim}\label{claim: no sieving-based quantum ISD}
    Consider a unique decoding instance of DP with parameters $n, k, w$. 
    For all $n', w', N \in \N$ satisfying Equation~\eqref{eq: bounds on N}, there is no (classical or quantum) algorithm for the oracle $\mathcal{A}$ in Algorithm~\ref{alg: ISD} such that the resulting quantum ISD algorithm has a better runtime than quantum Prange \cite{Bernstein10}. 
\end{claim}

\textit{Justification of the claim.} 
A trivial lower bound on the runtime of \textit{any} classical or quantum algorithm for the oracle $\mathcal{A}$ in Algorithm~\ref{alg: ISD} (with corresponding parameters $n', w', N$) is given by $N$, since the framework requires it to output $N$ solutions.  
Thus, given any such algorithm $\mathcal{A}$ and parameters $n', w', N$ satisfying Equation~\eqref{eq: bounds on N}, the runtime $T$ of the resulting quantum ISD algorithm must satisfy \begin{align*}
    T \geq \frac{N}{\sqrt{p_1 p_2}} &= \sqrt{\frac{ N {n \choose w}}{2^{n' - k} {n - n' \choose w - w'}}}
\end{align*}
where $p_1 := \frac{{n' \choose w'} {n - n' \choose w - w'}}{{n \choose w}}$ and $p_2 := \frac{N 2^{n' - k}}{{n' \choose w'}}$ as in \Cref{thm: complexity of quantum ISD assuming existence oracle}. 
For fixed $n', w'$, this lower bound on $T$ is minimal when $N$ is as small as possible, i.e., when $N = \Theta(N_{n', w'})$ for $N_{n',w'} := {n' \choose w'}/\left({w' \choose \frac{w'}{2}} {n'-w' \choose \frac{w'}{2}}\right)$. (In the remainder of the argument, we will leave out constant factors and assume for simplicity that $N = N_{n', w'}$.)
It follows that $T$ is lower bounded by \begin{align}
    \min_{(n', w', N) \text{ satisfying } \eqref{eq: bounds on N}} \sqrt{\frac{ N {n \choose w}}{2^{n' - k} {n - n' \choose w - w'}}} = \min_{(n', w') \text{ s.t. } {w' \choose \frac{w'}{2}} {n'-w' \choose \frac{w'}{2}} \geq 2^{n' - k}} \sqrt{\frac{ {n' \choose w'} {n \choose w}}{{w' \choose \frac{w'}{2}} {n'-w' \choose \frac{w'}{2}} 2^{n' - k} {n - n' \choose w - w'}}}. \label{eq: lower bound Quantum-SievingISD to be optimized}
\end{align} 

\noindent We want to show that this is never better than the runtime of quantum Prange, i.e., $\sqrt{\frac{{n \choose w}}{{n - k \choose w}}}$.

Numerically, we minimized Equation~\eqref{eq: lower bound Quantum-SievingISD to be optimized}; our code is publicly available in the aforementioned GitHub repository. We obtain that, for all  $n', w'$ satisfying ${w' \choose \frac{w'}{2}} {n'-w' \choose \frac{w'}{2}} \geq 2^{n' - k}$, it asymptotically holds that \begin{align*}
    \min_{(n', w') \text{ s.t. } {w' \choose \frac{w'}{2}} {n'-w' \choose \frac{w'}{2}} \geq 2^{n' - k}} \sqrt{\frac{ {n' \choose w'} {n \choose w}}{{w' \choose \frac{w'}{2}} {n'-w' \choose \frac{w'}{2}} 2^{n' - k} {n - n' \choose w - w'}}} \geq \sqrt{\frac{{n \choose w}}{{n - k \choose w}}}.
\end{align*}   
(Note here that the constraint ${w' \choose \frac{w'}{2}} {n'-w' \choose \frac{w'}{2}} \geq 2^{n' - k}$ is induced by Equation~\eqref{eq: bounds on N}.)
\qed

In fact, the optimal values obtained through numerical optimization essentially correspond to the values characterizing quantum Prange: $n' = k$ and $w' = 0$.

\begin{remark}[Comparison with the classical setting]
    The same argument does not apply to \textit{classical} SievingISD (non-surprisingly, since \cite{DEEK23} have shown that it indeed outperforms classical Prange). There, a trivial lower bound on the runtime of the oracle $\mathcal{A}$ (for any parameters $n', k, w', N$) would again be $N$, giving a (possibly non-tight) lower bound on the time complexity of classical SievingISD of $N/(p_1p_2)$. Unlike in the quantum setting, the dependency (and, in particular, the lower bound) on $N$ disappeared as it is canceled out by $p_2$, namely $N/(p_1p_2) = {n \choose w}/\left({n - n' \choose w - w'} 2^{n' - k} \right)$ since $p_2$ is linear in $N$. 
\end{remark}

\subsection{On Overcoming the Limitations of Quantum SievingISD}

The previous section illustrates limitations of the presented quantum SievingISD framework and implies that the framework should be adapted to outperform quantum Prange \cite{Bernstein10}. Specifically, the main bottleneck appears to be the \textit{lower bound} on the output size $N$ of the sieving subroutine imposed by \Cref{lem: lower bound on N for NNS} (where we emphasize that in this section we use $n', w'$ to specify the dimension and weight). Recall that the factors $T_{\mathcal{A}}$ and $p_2$ in the expected runtime $\Tilde{O}(T_{\mathcal{A}}/\sqrt{p_1p_2})$ of a quantum ISD algorithm both depend on $N$.  
Here, we present two natural approaches for potentially overcoming these limitations and explain why neither of them works. 

\subsubsection{Approach 1: From Pairs to Tuples.}
In this paper, we were considering finding pairs $(\xv_1,\xv_2)$ of vectors in our input list $\mathcal{L}$ such that $|\xv_1 + \xv_2| = w'$ (called \textit{solution pairs}). To guarantee the existence of $N = |\mathcal{L}|$ solution pairs, \Cref{lem: lower bound on N for NNS} implied that we need the lower bound on $N$. 
Inspired by lattice-based cryptanalysis, an idea to obtain a smaller lower bound on $N$ is to focus instead on finding $t$-tuples (for some $t > 2$).\footnote{The intuition comes from tuple-sieving algorithms (e.g., \cite{AC:KMPM19}) which consider $t>2$ to reduce the lower bound on $N$.}
That is, the sieving algorithm is instructed to (repeatedly) find $N$ tuples $(\xv_1,\dots, \xv_t) \in \mathcal{L}^t$ satisfying $|\xv_1 + \ldots + \xv_t| = w'$ for a given list $\mathcal{L}$ of size $N$. We refer to such tuples as \textit{$t$-solutions}. Note that, for $t = 2$, we recover the original setting in which the algorithm searches for solution pairs. 

Specifically, for $t \geq 2$, we say that a  \textit{$t$-sieving algorithm} is an algorithm constructed similarly to the 2-sieving algorithm in Algorithm~\ref{alg: sieving}, but instead searching for $t$-solutions: in each iteration $i$ of the sieving part, it aims to find $N$ $t$-solutions given a list $\mathcal{L}$ of $N$ independent and uniformly random vectors from $\Ss_{w'}^{n'}$. Writing $p^{(t)}$ for the probability that a $t$-tuple of independent and uniformly random vectors in $\Ss_{w'}^{n'}$ forms a $t$-solution, 
the expected number of $t$-solutions in $\mathcal{L}$ is then given by $N^t p^{(t)}$.  
To ensure that, on average, there exist at least $N$ $t$-solutions, the list size $N$ thus needs to satisfy $N^t p^{(t)} \geq N$, possibly resulting in a reduced lower bound on the output size $N$.

We refer to \textit{quantum $t$-sieving ISD} as the resulting quantum ISD algorithm where we instantiate the oracle $\mathcal{A}$ with a $t$-sieving algorithm, where we recall that $\mathcal{A}$ aims to find $N$ solutions to $\DP(n',k,w')$.  
We keep the meaning of $p_1$ and $p_2$ from \Cref{thm: complexity of quantum ISD assuming existence oracle}, and write $p_2 = N q_2$ to highlight that $p_2$ depends on the output size $N$. (Note that $p_1$ and $q_2$ do not depend on $N$ or $t$.) The expected runtime of quantum $t$-sieving ISD is then $\Tilde{O}\left(T_{\mathcal{A}}/\sqrt{p_1 q_2 N}\right)$, where $T_{\mathcal{A}}$ is the expected runtime of $\mathcal{A}$. 
\

However, even if the use of $t$-tuples for $t > 2$ might potentially reduce the lower bound on the output size $N$, we obtain the following lower bound on the runtime of quantum $t$-sieving ISD. 
We remark that the assumption holds, for instance, if the runtime of any $t$-sieving algorithm is lower bounded by the optimal runtime of $2$-sieving (since $1/p^{(2)}$ is a lower bound on the output size of a 2-sieving algorithm). 

\begin{proposition}[A lower bound on quantum $t$-sieving ISD] 
Assume that the runtime of any $t$-sieving algorithm is at least $1/p^{(2)}$. Consider a quantum ISD algorithm with the aforementioned $t$-sieving algorithm as oracle $\mathcal{A}$. The expected runtime of this algorithm is $\Omega(1/\sqrt{p_1 q_2 p^{(2)}})$.
\end{proposition}

By \Cref{sec: limitations of quantum SievingISD}, the latter is (asymptotically) never better than the runtime of quantum Prange. Hence, we can conclude that, under the stated assumption, quantum $t$-sieving ISD does no better than quantum Prange for all $t \geq 2$.  

\begin{proof}
We analyze the cases $N \leq 1/p^{(2)}$ and $N \geq 1/p^{(2)}$ separately, where   $N$ (as usual) denotes the output size of~$\mathcal{A}$. 
First, suppose that $N \leq 1/p^{(2)}$. 
By the assumption, the runtime $T_{\mathcal{A}}$ of any $t$-sieving algorithm is at least $1/p^{(2)}$, so the  runtime of the resulting quantum ISD algorithm is lower bounded by $T_{\mathcal{A}}/\sqrt{p_1 q_2 N} \geq 1/ \left(p^{(2)} \sqrt{p_1 q_2 N}\right) \geq 1/\sqrt{p_1 q_2 p^{(2)}}$.
On the other hand, if $N \geq 1/p^{(2)}$, then $T_{\mathcal{A}}/\sqrt{p_1 q_2 N} \geq \sqrt{N}/\sqrt{p_1 q_2} \geq  1/\sqrt{p_1 q_2 p^{(2)}}$ since $T_{\mathcal{A}} \geq N$.  
(Note that we omit writing $\Tilde{O}(\cdot)$ throughout the proof.)
\qed 
\end{proof}

\subsubsection{Approach 2: Varying List Sizes.} 

The lower bound on the output size $N$ of the sieving subroutine comes from maintaining the same list size $N$ throughout each iteration in the sieving algorithm. A natural question to ask is whether it would be possible to vary the list size in order to overcome this intrinsic limitation of quantum SievingISD.

We propose the following approach for varying list sizes. Suppose that for given parameters $n', w'$ there exists a (classical/quantum) algorithm for the subroutine $\mathcal{A}$ in the quantum ISD algorithm that iteratively applies a sieving step of the following form, where the values $N_i$ are arbitrary. It~starts by sampling a list $\mathcal{L}_0$ of $N_0$ independent and uniformly random vectors from $\Ss_{w'}^{n'}$. 
For iteration~$i = 1$ up to $i = n'-k$, the algorithm finds $N_i$ pairs $(\xv_1, \xv_2) \in \mathcal{L}_{i-1}^2$ satisfying $|\xv_1 + \xv_2| = w'$, yielding a new list $\mathcal{L}_{i}$ of size $N_i$. 
Note that the expected number of solution pairs is $N_{i-1}^2 p$, where $p := {w' \choose w'/2} {n' - w' \choose w'/2}/{n' \choose w'}$ is the probability that a uniformly random pair is a solution pair (previously denoted by $p^{(2)}$). 
Therefore, we will only consider $N_i$ satisfying $N_{i} \leq N_{i-1}^2 p$. 

\begin{remark}
    If $N_i = N_0$ for all $i$, then we get the same lower bound on the list size as we considered so far, namely $1/p$, resulting in the original setting. We now instead consider the situation where the $N_i$'s might differ. 
\end{remark}

We now sketch the proof that there is no choice of $N_0, \ldots, N_{n'-k}$ satisfying $N_{i} \leq N_{i-1}^2 p$ for which this algorithm yields runtime better than that of quantum Prange (Lemma~\ref{lem: quantum Prange}). 

We start by observing that the overall runtime of this quantum ISD algorithm is essentially $T := T_{\max}/\sqrt{p_1 q_2 N_{n'-k}}$, where $T_{\max}$ is the maximum among the runtimes of the iterations, and $p_1$ and $p_2 = q_2 N_{n'-k}$ are the probabilities from \Cref{thm: complexity of quantum ISD assuming existence oracle}. (Note that the output size $N_{n'-k}$ was previously denoted by $N$.) Since $T_{\max} \geq N_{\max} := \max_{i \geq 0} N_{i}$, we have $T \geq N_{\max}/\sqrt{p_1 q_2 N_{n'-k}}$. 
Recall from \Cref{sec: limitations of quantum SievingISD} that, for all allowed $(n',w')$, $1/\sqrt{p_1 q_2 p}$ is asymptotically lower bounded by the runtime of quantum Prange. Therefore, it suffices to show that there is no choice of $N_0, \ldots, N_{n'-k}$ such that $N_{\max}/\sqrt{N_{n'-k}} < 1/\sqrt{p}$. 

Suppose for contradiction there is a choice of $N_0, \ldots, N_{n'-k}$ satisfying $N_{i} \leq N_{i-1}^2 p$ for all $i \geq 1$ and $N_{\max}/\sqrt{N_{n'-k}} < 1/\sqrt{p}$.
We analyze the two cases $N_{n'-k} \geq \frac{1}{p}$ and $N_{n'-k} < \frac{1}{p}$ separately. 
We start with the case $N_{n'-k} \geq 1/p$. Then 
$N_{\max}/\sqrt{N_{n'-k}} \geq \sqrt{N_{n'-k}} \geq 1/\sqrt{p}$.

It remains to consider the case $N_{n'-k} < 1/p$. 
If $N_0 \geq 1/p$, then $N_0 > N_{n'-k}$, so $N_{\max}/\sqrt{N_{n'-k}} \geq N_0/\sqrt{N_{n'-k}} > \sqrt{N_0} \geq 1/\sqrt{p}$. 
Finally, consider the case that $N_0 < 1/p$. 
Note that $N_{i} \leq N_{i-1}^2 p$ for all $i\geq 1$ implies that
$N_{i} \leq N_0^{2^i} p^{2^i - 1}$  for all $i \geq 0$. 
In particular, the size of the output list satisfies $N_{n'-k} \leq N_0^{2^{n'- k}} p^{2^{n'- k} - 1}$. 
Therefore, we obtain that \begin{align*}
    \frac{N_{\max}}{\sqrt{N_{n'-k}}} \geq \frac{N_0}{\sqrt{N_{n'-k}}} \geq 
    \sqrt{\frac{N_0^2}{N_0^{2^{n'- k}} p^{2^{n'- k} - 1}}} = \frac{1}{\sqrt{p}} \frac{1}{\sqrt{N_0^{2^{n'- k} - 2} p^{2^{n'- k} - 2}}} > \frac{1}{\sqrt{p}}
\end{align*}
since $N_0 p < 1$. That is, in all possible cases we showed that $N_{\max} /\sqrt{N_{n'-k}} \geq 1/\sqrt{p}$, so we reached a contradiction. We conclude that there is no suitable choice for the $N_i$ that enables to outperform quantum Prange.

\bibliographystyle{alpha}
\bibliography{references_all}

\end{document}